\documentclass[11pt]{amsart}

\usepackage{amssymb,latexsym, amsmath, amsxtra,amsfonts,amsthm}
\usepackage{braket}
\usepackage{caption}
\usepackage{cite}
\usepackage{cleveref}
\usepackage[usenames, dvipsnames]{color} 
\usepackage{dsfont}
\usepackage{empheq}
\usepackage{epsfig}
\usepackage{etoolbox}
\usepackage{float}
\usepackage{mathtools}
\usepackage{mdframed}
\usepackage{mleftright}
\usepackage{physics}
\usepackage{stmaryrd}
\usepackage{subcaption}
\usepackage{thmtools,thm-restate}
\usepackage{tikz}
\usepackage{times}
\usepackage{url}
\usepackage{bm}
\usepackage[]{graphics}
\usepackage{bbm}
\usepackage[normalem]{ulem}
\usetikzlibrary{arrows,automata,matrix,positioning,calc,shapes,backgrounds,trees}

\newtheorem*{thm}{Theorem}
\newtheorem*{lemmum}{Lemma}
\newtheorem{theorem}{Theorem}[section]

\newtheorem{lemma}[theorem]{Lemma}

\newtheorem{proposition}[theorem]{Proposition}

\let\emptyset\varnothing

\DeclareMathOperator{\RE}{Re}

\DeclareMathOperator{\Sgn}{sgn}

\DeclareMathOperator{\mom}{MoM}
\DeclareMathOperator{\Sc}{Sc}

\begin{document}
\title{On the moments of the moments of the characteristic polynomials of random unitary matrices}
\abstract Denoting by $P_N(A,\theta)=\det(I-Ae^{-i\theta})$ the characteristic polynomial on the unit circle in the complex plane of an  $N\times N$ random unitary matrix $A$, we calculate the $k$th moment, defined with respect to an average over $A\in U(N)$, of the random variable corresponding to the $2\beta$th moment of $P_N(A,\theta)$ with respect to the uniform measure  $\frac{d\theta}{2\pi}$, for all $k, \beta\in\mathbb{N}$ .  These moments of moments have played an important role in recent investigations of the extreme value statistics of characteristic polynomials and their connections with log-correlated Gaussian fields.  Our approach is based on a new combinatorial representation of the moments using the theory of symmetric functions, and an analysis of a second representation in terms of multiple contour integrals.  Our main result is that the moments of moments are polynomials in $N$ of degree $k^2\beta^2-k+1$.  This resolves a conjecture of Fyodorov \& Keating~\cite{fyodorov14} concerning the scaling of the moments with $N$ as $N\rightarrow\infty$, for $k,\beta\in\mathbb{N}$.  
Indeed, it goes further in that we give a method for computing these polynomials explicitly and obtain a general formula for the leading coefficient.
\endabstract

\author{E. C. Bailey \& J. P. Keating} \address{School of Mathematics,
University of Bristol, Bristol, BS8 1TW, United Kingdom}
\email{emma.bailey@bristol.ac.uk, j.p.keating@bristol.ac.uk}

\maketitle

%%%%%%%%%%%%%%%%%%%
%                                                     %
%                                                     %
%             Introduction                     %
%                                                     %
%                                                     %
%%%%%%%%%%%%%%%%%%%
\section{Introduction}\label{sec:intro}

Let
\begin{equation}\label{eq:charpoly}
P_N(A,\theta)=\det(I-Ae^{-i\theta}),
\end{equation}
denote the characteristic polynomial of an $N\times N$ unitary matrix $A$ on the unit circle in the complex plane.  The typical values taken by $P_N$ when $A$ is chosen at random, uniformly with respect to Haar measure on the unitary group $U(N)$ (i.e.~from the Circular Unitary Ensemble of Random Matrix Theory), have been the subject of extensive study.  The moments of $P_N$ and its logarithm were computed in \cite{keasna00a} using the Selberg integral and compared with the corresponding moments of the Riemann zeta function, $\zeta(s)$, on its critical line (${\rm Re}s=1/2$).  It follows from these calculations that $\log P_N(A,\theta)/\sqrt{\tfrac{1}{2}\log N}$ satisfies a central limit theorem when $N\rightarrow\infty$, in that the real and imaginary parts independently converge to normal random variables with zero mean and unit variance.  This is true as well without normalising, in a distributional sense \cite{HKOC}.     The correlations of $\log |P_N(A,\theta)|$ can be computed using, for example, formulae due to Diaconis and Shahshahani \cite{DiacShahsh}, and shown to satisfy
\begin{equation}
\mathbb{E}_{A\in U(N)}\left(\log |P_N(A,\theta)|\log |P_N(A,\theta+x)|\right)\sim\begin{cases}\frac{1}{2}\log N&|x|<<\frac{1}{N}\\
-\frac{1}{2}\log |x|&1>>|x|>>\frac{1}{N}\end{cases}
\end{equation}
when $N\rightarrow\infty$.  (The imaginary part of $\log P_N(A,\theta)$ exhibits similar behaviour.)

The fact that $\log |P_N(A,\theta)|$ behaves like a log-correlated Gaussian random function has stimulated a good deal of interest recently, as it suggests a connection with other similar random fields such as those associated with the Branching Random Walk, Branching Brownian Motion, the 2-dimensional Gaussian Free Field, and Liouville quantum gravity.  This observation, together with heuristic calculations and numerical experiments (c.f.~\cite{FGK}), motivated a series of conjectures~\cite{fyodorov12, fyodorov14} concerning the maximum of $|P_N(A,\theta)|$ on the unit circle,
 \begin{equation}
 P_{\rm max}(A)=\max_{0\le \theta <2\pi}|P_N(A,\theta)|.
 \end{equation}
 
The heuristic calculations described in \cite{fyodorov14} are based on an analysis of the random variable
\begin{equation}
\label{mo}
Z_N(A, \beta)\coloneqq\frac{1}{2\pi}\int_0^{2\pi}|P_N(A,\theta)|^{2\beta}d\theta
\end{equation}
which is the $2\beta$th moment of $|P_N(A,\theta)|$ with respect to the uniform measure on the unit circle $\frac{d\theta}{2\pi}$.  Specifically, the calculations centre on computing the moments of this random variable with respect to an average over $A\in U(N)$:
\begin{equation}\label{eq:mom}
\mom_N(k,\beta)\coloneqq \mathbb{E}_{A\in U(N)}\left(\left(\frac{1}{2\pi}\int_0^{2\pi}|P_N(A,\theta)|^{2\beta}d\theta\right)^k\right).
\end{equation} 
We refer to the latter as the {\it moments of the moments} of $P_N(A,\theta)$.  They will be the main focus of our attention.  

We also note that the integrand in (\ref{mo}), when appropriately normalised, 
\begin{equation}
\frac{|P_N(A,\theta)|^{2\beta}}{\mathbb{E}|P_N(A,\theta)|^{2\beta}}\frac{d\theta}{2\pi}
\end{equation} 
has been the subject of considerable interest because it has been proved \cite{Webb, NSW} to converge to a limiting Gaussian multiplicative  chaos measure \cite{Kahane1, BerLectures, RVLectures} for $\beta\in (-\frac{1}{4}, 1)$ (c.f.~\cite{SW} for a corresponding result for the Riemann zeta-function on the critical line).  Importantly, there is expected to be a freezing transition \cite{fyodorov14}  at $\beta=1$, leading to a different regime of behaviour when $\beta>1$.

One of the main conjectures of \cite{fyodorov14} is that when $N\rightarrow\infty$
\begin{equation}\label{conj}
\mom_N(k,\beta)\sim
\begin{cases}
\left(\frac{\left(G(1+\beta)\right)^2}{G(1+2\beta)\Gamma(1-\beta^2)}\right)^k
\Gamma(1-k\beta^2)N^{k\beta^2}
&k<1/\beta^2
\\
c(k, \beta)N^{k^2\beta^2-k+1}
&k>1/\beta^2
\end{cases}
\end{equation}
where $G(s)$ is the Barnes $G$-function and $c(k, \beta)$ is an unspecified function of $k$ and $\beta$\footnote{By $A(N)\sim B(N)$, we mean that $A(N)/B(N)\to 1$ when $N\rightarrow\infty$.}.  At the transition point $k=\beta^2$, one should expect that the moments of moments grow like $N\log N$.  One justification for this conjecture follows from a heuristic calculation of the moments when $k$ is an integer \cite{fyodorov12, fyodorov14, Keating_lec}, which is based on the fact that for $k\in\mathbb{N}$
\begin{equation}
\label{basic}
\mom_N(k,\beta)=\frac{1}{(2\pi)^k}\int_0^{2\pi}\dots\int_0^{2\pi}
{\mathbb{E}}\prod_{j=1}^k|P_N(A,\theta_j)|^{2\beta}d\theta_j.
\end{equation}
The integrand in (\ref{basic}) can be computed asymptotically when $N\rightarrow\infty$ and the $\theta_j$s are fixed and distinct using the appropriate Fisher-Hartwig formula \cite{FF}.  The resulting integrals over the $\theta_j$s can then be computed when $k<1/\beta^2$ using the Selberg integral, leading to the expression in the conjecture (\ref{conj}) in this range.  This expression diverges as $k$ approaches $1/\beta^2$ from below.  The reason for this is that when $k\ge1/\beta^2$, singularities associated with coalescences of the  $\theta_j$s become important.  Developing a precise asymptotic in the range $k\ge1/\beta^2$ therefore requires a Fisher-Hartwig formula that is valid uniformly as the Fisher-Hartwig singularities coalesce, and achieving this in general is an important open problem.  From this perspective, the regime $k\ge1/\beta^2$ is the more challenging one.   

When $k=2$, a uniform Fisher-Hartwig asymptotic formula was established by Claeys and Krasovsky \cite{CK}, who used this to prove the powers of $N$ appearing in (\ref{conj}) for all $\beta$ and to relate $c(2, \beta)$ to a particular Painlev\'e transcendent.  

In a closely analogous problem in which $\log |P_N(A,\theta)|$ (c.f.~(\ref{Fourier})) is replaced by a random Fourier series with the same correlation structure -- such series can be considered as one-dimensional models of the two-dimensional Gaussian Free Field -- the analogue of conjecture (\ref{conj}), due to Fyodorov and Bouchaud \cite{FB2008}, has recently been proved in the regime $k<1/\beta^2$ for all $k$ and $\beta$ by Remy \cite{Remy1} using ideas from conformal field theory \cite{KRV}.

We note that the conjecture described above extends to the other circular ensembles (i.e.~to the C$\beta$E) \cite{keasna00a, CMN, FGK} and to the Gaussian ensembles \cite{FyoKhorSimm, FyoSimm, FLD}.  We note as well that there are extensive mathematics and physics literatures on log-correlated Gaussian fields; see, for example \cite{DRZ}, \cite{FLD} and \cite{CLD}, and references contained therein.  There has been a particular focus on the freezing transition at $\beta=1$.  In the case of uncorrelated Gaussian fields -- known as the Random Energy Model -- this is well understood; see, for example, \cite{Derrida, Kistler}.  For log-correlated fields the freezing transition continues to be a focus of research; see, for example, \cite{fyodorov14, SZ} and references therein.  

Our focus here will be on the conjecture for the asymptotics of the moments of moments (\ref{conj}) when $k\in\mathbb{N}$ and $\beta\in\mathbb{N}$.  Note that this immediately places us in the regime where $k\beta^2\ge 1$, and so in the more difficult regime which is dominated by coalescing Fisher-Hartwig singularities, and where progress has been limited thus far to the cases of $k=1,2$. Here one can exploit connections with representation theory and integrable systems that have not been incorporated in the probabilistic approaches taken previously.  Specifically, we shall use three different, but equivalent, exact (rather than asymptotic) expressions for the integrand in (\ref{basic}).  This allows us to circumvent the problems described above associated with coalescing Fisher-Hartwig singularities.  We also note that our results include the freezing transition point at $\beta=1$.  

The first of these expressions, which takes the form of a combinatorial sum and was proved in \cite{cfkrs1}, enables us to compute $\mom_N(k,\beta)$ exactly and explicitly for small values of $k$ and $\beta$, when both take values in $\mathbb{N}$.  This suggests a refinement of conjecture (\ref{conj}) in this case:
\begin{equation}\label{conj:FK}
\mom_N(k,\beta)= \operatorname{Poly}_{k^2\beta^2-k+1}(N),
\end{equation}
where $\operatorname{Poly}_{k^2\beta^2-k+1}(N)$ is a polynomial in the variable $N$ of degree $k^2\beta^2-k+1$.  This obviously implies (\ref{conj}) in the range $k\ge1/\beta^2$ for $k, \beta\in\mathbb{N}$.  We present the calculation in an Appendix, where we give explicit examples of the polynomials that arise.  This method can be used to establish that $\mom_N(k,\beta)$ is in general a polynomial in $N$, but does not straightforwardly determine the order of the polynomial in question.

We then go on to prove (\ref{conj:FK}) using two alternative approaches.  The first of these uses a second formula for the integrand in (\ref{basic}) that is based on the representation theory of the unitary group and involves expressing $\mom_N(k,\beta)$ in terms of a sum of semistandard Young tableaux via the theory of symmetric functions.  The application of the theory of symmetric functions in this context was developed by Bump and Gamburd \cite{bumgam06}, who used it to analyse the moments of characteristic polynomials, following \cite{keasna00a} and \cite{cfkrs1}.  It allows us to prove that $\mom_N(k,\beta)$ is bounded by a polynomial function of $N$ of degree less than or equal to $k^2\beta^2$ at integer values of $k,\beta$, and $N$.  The other approach involves a third formula for the integrand in (\ref{basic}), which takes the form of a multiple contour integral and which was also proved in \cite{cfkrs1}.  This allows us to compute the large-$N$ asymptotics of $\mom_N(k,\beta)$, using methods developed in \cite{keaodg08, krrr15}.  We show in this way that $\mom_N(k,\beta)$ is an analytic function of $N$ that grows like $N^{k^2\beta^2-k+1}$ as $N\rightarrow\infty$.  This approach allows us to obtain a formula for the leading coefficient of the polynomial in (\ref{conj:FK}), which corresponds to evaluating the function $c(k,\beta)$ in (\ref{conj}) when $k$ and $\beta$ are both integers.  Combining these various results allows us to deduce that $\mom_N(k,\beta)$ is a polynomial in $N$ of order $k^2\beta^2-k+1$, thereby proving (\ref{conj:FK}).

The fact that $\mom_N(k,\beta)$ is a polynomial in the variable $N$ when $k$ and $\beta$ both take values in $\mathbb{N}$ means that in this case we have an exact formula.  This is a consequence of this problem being integrable, as is clear from the analysis based on symmetric functions.  From the perspective of asymptotics, it means that we know the complete structure of the asymptotics of  $\mom_N(k,\beta)$; that is, we know the general form of all terms in the asymptotic expansion, not just the leading order term.

We emphasize that our main motivation here is to prove (\ref{conj}), and in particular its refinement (\ref{conj:FK}), in the regime $k\beta^2\ge 1$ where previous approaches have failed in general (i.e.~other than when $k=2$) because they require a general Fisher-Hartwig formula valid as $k$ singularities coalesce.  Our approach circumvents this obstacle.

This paper is structured as follows.  In the next subsection we state some formulae for $\mom_N(k,\beta)$ that can be obtained straightforwardly from expressions already in the literature and formulate our general results as theorems.  In the Appendix we calculate $\mom_N(k,\beta)$ for small values of $k$ and $\beta$, motivating (\ref{conj:FK}).  In Section 2, we explain the calculation involving symmetric functions, and then in Section 3 we describe the calculation involving multiple integrals.  {In Section 4 we discuss some connections between our main result and approaches to analysing rigorously the value distribution of $P_{\rm max}(A)$, in the context of the conjectures made in~\cite{fyodorov12, fyodorov14}, as well as setting out some thoughts on potential extensions and applications, including to moments of the Riemann zeta-function and other $L$-functions in short intervals, as well as to other random matrix ensembles.}

%%%%%%%%%%%%%%%%%%%
%                                                     %
%                                                     %
%           Known Results                  %
%                                                     %
%                                                     %
%%%%%%%%%%%%%%%%%%%

\subsection{Results for $\mom_N(1,\beta)$ and $\mom_N(2,\beta)$, for $\beta\in\mathbb{N}$.}\label{knownresults}

We set out in this subsection some results concerning $\mom_N(k,\beta)$ that can be obtained straightforwardly from calculations in the literature and that prove (\ref{conj:FK}) when $k=1$ and $k=2$. 

The case $k=1,\beta\in \mathbb{N}$ follows immediately from the moment formula of Keating and Snaith~\cite{keasna00a} (c.f.~also \cite{BF}), and matches with the conjecture.   Specifically, 
\begin{equation}
\label{k=1}
\mom_N(1,\beta)={\mathbb{E}}|P_N(A,\theta)|^{2\beta}=\prod_{0\le i,j\le \beta-1}\left(1+\frac{N}{i+j+1}\right),
\end{equation}
which is clearly a polynomial in $N$ of degree $\beta^2$.  In this case the leading order coefficient can be calculated \cite{keasna00a} to be
\begin{equation}
\label{leading}
\prod_{j=0}^{\beta-1}\frac{j!}{(j+\beta)!}.
\end{equation}
The calculation of the average in (\ref{k=1}) was carried out in \cite{keasna00a} using the Weyl integration formula and Selberg's integral.  Bump and Gamburd~\cite{bumgam06} later give an alternative proof using symmetric function theory.  In this second approach, the expression (\ref{leading}) was obtained by counting certain semistandard Young tableaux.  We shall see these parallel stories of symmetric function theory and complex analysis continuing for higher values of $k$.  

A proof of (\ref{conj:FK}) when $k=2, \beta\in\mathbb{N}$ follows directly from formulae given in \cite{krrr15} {(and differs from the proof given by Claeys and Krasovsky~\cite{CK} which proves \eqref{conj} for all $\beta$, but does not identify the polynomial structure when $\beta\in\mathbb{N}$)}. Recall that for $A\in U(N)$, the \textit{secular coefficients} of $A$ are the coefficients of its characteristic polynomial 
\begin{equation}\label{def:seccoeff}
\det(I+xA)=\sum_{n=0}^N\Sc_n(A)x^n.
\end{equation}
The following theorem is proved in \cite{krrr15} (theorem 1.5 in that paper).
\begin{theorem}\label{thm:krrr1}
For $A\in U(N)$, define
\begin{equation}
I_\eta(m;N)\coloneqq\int_{U(N)}\Big|\sum_{\substack{j_1+\cdots+j_\eta=m\\0\leq j_1,\dots,j_\eta\leq N}}\Sc_{j_1}(A)\cdots \Sc_{j_\eta}(A)\Big|^2dA.
\end{equation}
If $c=m/N, c\in[0,\eta]$, then $I_\eta(m;N)$ is a polynomial in $N$ and
\begin{equation}\label{eq:krrr1}
I_\eta(m;N)=\gamma_\eta(c)N^{\eta^2-1}+O_\eta(N^{\eta^2-2}),
\end{equation}
where 
\begin{equation}\label{krrr:leadingcoeff}
\gamma_\eta(c)=\sum_{0\leq l<c}\binom{\eta}{l}^2(c-l)^{(\eta-l)^2+l^2-1}p_{\eta,l}(c-l),
\end{equation}
with $p_{\eta,l}(c-l)$ being polynomials in $(c-l)$. 
\end{theorem}
With a change of variables, it can easily be seen that \cref{thm:krrr1} proves (\ref{conj:FK}) when $k=2$, $\beta\in \mathbb{N}$.  We make use of the generating series for $I_\eta(m,N)$ given in~\cite{krrr15}, 
\begin{equation}\label{generatingseries}
\sum_{0\leq m\leq \eta N}I_\eta(m;N)x^m=\int_{U(N)}\det(I-A)^\eta\det(I-xA^*)^\eta dA.
\end{equation}
Then we see that 
\begin{align*}
\mom_N(2,\beta)&=\frac{1}{(2\pi)^2}\int_0^{2\pi}\int_0^{2\pi}\mathbb{E}\left(|P_N(A,\theta_1)|^{2\beta}|P_N(A,\theta_2)|^{2\beta}\right)d\theta_1d\theta_2\\
&=\frac{1}{(2\pi)^2}\int_0^{2\pi}\int_0^{2\pi}\int_{U(N)}\left|\det(1-Ae^{-i\theta_1})\right|^{2\beta}\left|\det(1-Ae^{-i\theta_2})\right|^{2\beta}dAd\theta_1d\theta_2\\
&=\frac{1}{(2\pi)^2}\int_0^{2\pi}\int_0^{2\pi}e^{i\beta(\theta_2-\theta_1)N}\sum_{0\leq m\leq 2\beta N}I_{2\beta}(m;N)e^{i(\theta_1-\theta_2)m}d\theta_1d\theta_2\\
&=\sum_{0\leq m\leq 2\beta N}I_{2\beta}(m;N)\delta_{m-\beta N}.
\end{align*}
Immediately, \cref{thm:krrr1} gives us that $\mom_N(2,\beta)$ is a polynomial in $N$, and we have the correct leading order, 
\begin{equation}\mom_N(2,\beta)\sim \gamma_{2\beta}(\beta)N^{4\beta^2-1}+O_{2\beta}(N^{4\beta^2-2}),
\end{equation}
 provided that $\gamma_{2\beta}(\beta)\neq 0$.  \Cref{thm:krrr1} was proved by two methods: symmetric function theory and complex analysis. The former determines an equivalent structure for $\gamma_\eta(c)$ to that given in \eqref{krrr:leadingcoeff} coming from a standard lattice point count, which proves that $I_\eta(m;N)$ is a polynomial in $N$ and makes it clear that $\gamma_{2\beta}(\beta)\neq 0$. By using complex analysis the result regarding the leading order in $N$ can be established, and the form for $\gamma_\eta(c)$ given in \eqref{krrr:leadingcoeff} is found.

%%%%%%%%%%%%%%%%%%%
%                                                     %
%                                                     %
%          Results and Outline            %
%                                                     %
%                                                     %
%%%%%%%%%%%%%%%%%%%

\subsection{Results}

Our approach combines the methods and formulae developed in \cite{keasna00a, cfkrs1, bumgam06, keaodg08, krrr15}; in particular we make use of the complex analytic techniques employed in the latter two papers.  We first reformulate (\ref{conj:FK}) in terms of symmetric function theory and a lattice point count function.  This gives a polynomial bound on $\mom_N(k,\beta)$ at integer values of $k$, $\beta$, and $N$.  We next use a representation in terms of multiple contour integrals; this furnishes an expression for $\mom_N(k,\beta)$ as an entire function of $N$ and allows us to prove the following theorem.

\begin{theorem}\label{thm:mom}
Let $k,\beta\in\mathbb{N}$.  Then 
\begin{equation}\label{eq:momkbeta}
\mom_N(k,\beta)=\gamma_{k,\beta}N^{k^2\beta^2-k+1}+O(N^{k^2\beta^2-k}),
\end{equation}
where $\gamma_{k,\beta}$ can be written explicitly in the form of an integral. 
\end{theorem}

Using a combinatorial sum equivalent to the multiple contour integrals due to~\cite{cfkrs1}, we then deduce the following result.

\begin{theorem}\label{thm:polynomial}
Let $k, \beta\in\mathbb{N}$. Then $\mom_N(k,\beta)$ is a polynomial in $N$.
\end{theorem}

These theorems together prove (\ref{conj:FK}) for $k, \beta\in\mathbb{N}$. 

\subsection{Acknowledgements}  We thank Edva Roditty-Gershon and Scott Harper for helpful discussions, and Euan Scott for contributing to preliminary computations of the moments using the method outlined in the Appendix.  ECB is grateful to the Heilbronn Institute for Mathematical Research for support.  JPK is pleased to acknowledge support from a Royal Society Wolfson Research Merit Award and ERC Advanced Grant 740900 (LogCorRM).  We are most grateful to the referees for their careful reading of the manuscript and for a number of helpful questions and suggestions.

%%%%%%%%%%%%%%%%%%%
%                                                     %
%                                                     %
%                    SFT                          %
%                                                     %
%                                                     %
%%%%%%%%%%%%%%%%%%%
\section{Symmetric Function Theory}

As with the cases $k=1,2$ and $\beta\in\mathbb{N}$, we can rephrase the problem in terms of symmetric function theory.   For an introduction to this topic, see \cite{macdonald98} and \cite{stanley99}.  For a self-contained review of the tools required for the following calculation, see~\cite{krrr15}. 

The aim of this section is twofold.  Firstly, we highlight the role that symmetric function theory plays in the analysis of the moments of moments as $k$ increases.  Secondly, in understanding how the results of Bump and Gamburd~\cite{bumgam06} and Keating et al.~\cite{krrr15} generalise for higher $k$, we recover an explicit polynomial bound on $\mom_N(k,\beta)$ at integer values of $k$, $\beta$, and $N$.

\begin{proposition}\label{prop:symmetric}
We have that
\begin{equation}
 \mathbb{E}_{A\in U(N)}\left(\prod_{j=1}^k|P_N(A,\theta_j)|^{2\beta}\right)=\frac{s_{\langle N^{k\beta}\rangle}(e^{i\underline{\theta}})}{\prod_{j=1}^ke^{iN\beta\theta_j}},
 \end{equation} 
 where $s_\lambda(x_1,\dots,x_n)$ is the Schur polynomial in $n$ variables with respect to the partition $\lambda$, and we write $\langle \lambda^n\rangle=(\overbrace{\lambda,\dots,\lambda}^n)$ and 
 \begin{equation}\label{schurvector}e^{i\underline{\theta}}=(\overbrace{e^{i\theta_1},\dots,e^{i\theta_1}}^\beta,\overbrace{e^{i\theta_2},\dots,e^{i\theta_2}}^\beta,\dots,\overbrace{e^{i\theta_{k}},\dots,e^{i\theta_{k}}}^\beta,\overbrace{e^{i\theta_1},\dots,e^{i\theta_1}}^\beta,\overbrace{e^{i\theta_2},\dots,e^{i\theta_2}}^\beta,\dots,\overbrace{e^{i\theta_{k}},\dots,e^{i\theta_{k}}}^\beta).\end{equation}
\end{proposition} 
%%%%%%
%% Currently ommitting the proof because of the similarity to Bump and Gamburd
%%%%%%
Hence, we can rewrite $\mom_N(k,\beta)$ in the following way
\begin{equation}\label{eq:symmetricmom}
\mom_N(k,\beta)=\frac{1}{(2\pi)^k}\int_0^{2\pi}\cdots\int_0^{2\pi}\frac{s_{\langle N^{k\beta}\rangle}(e^{i\underline{\theta}})}{\prod_{j=1}^ke^{iN\beta\theta_j}}\prod_{j=1}^kd\theta_j.
\end{equation}
In general, one can express a Schur function as a sum over all semistandard Young tableaux (SSYT) of shape $\lambda$, 
\begin{equation}
s_\lambda(x_1,\dots,x_n)=\sum_T\underline{x}^T=\sum_Tx_1^{t_1}\dots x_n^{t_n},
\end{equation} 
where $t_i$ is the number of times $i$ appears in a given tableau $T$.  Thus, in the situation above, we find 
\begin{equation}
s_{\langle N^{k\beta}\rangle}(e^{i\underline{\theta}})=\sum_T e^{i\theta_1\tau_1}\cdots e^{i\theta_k\tau_k},\end{equation}
where the sum is over all SSYT of rectangular shape with $k\beta$ rows by $N$ columns, and 
\begin{equation}
\tau_j=t_{2(j-1)\beta+1}+\cdots+t_{2j\beta},\quad\text{ for } j\in\{1,\dots,k\}.
\end{equation}
Hence we have
\begin{align}
\mom_N(k,\beta)&=\frac{1}{(2\pi)^k}\int_0^{2\pi}\dots\int_0^{2\pi}\sum_Te^{i\theta_1(\tau_1-N\beta)}\cdots e^{i\theta_k(\tau_k-N\beta)}\prod_{j=1}^kd\theta_j\\
&=\sum_{\widetilde{T}}1,
\end{align}
where the sum is now over $\widetilde{T}$, a set of restricted SSYT described as follows.  The Kronecker $\delta$-function arising from the integral over the $\theta_j$s imposes a further condition upon the rectangular SSYT: $\tau_j=N\beta$. That is, there have to be $N\beta$ entries from each of the sets 
\begin{equation}
\{2\beta(j-1)+1,\dots,2j\beta\},\quad\text{ for } j\in\{1,\dots,k\}.
\end{equation}
Thus we define \textit{restricted} SSYT (RSSYT) to be those SSYT, $\widetilde{T}$, satisfying this additional condition.  When specialised to the case of $k=1,\beta\in\mathbb{N}$, this approach matches the proof given by Bump and Gamburd (corollary 1 of~\cite{bumgam06}) which also uses the following well-known lemma, see for example~\cite{stanley71}.

\begin{lemma}~\label{lemma:ssyt}
The number of SSYT of shape $\lambda$ with entries in $1,2,\dots,n$  can be found by evaluating the Schur polynomial $s_\lambda(1,\dots,1)$.  We implicitly extend $\lambda$ with zeros until it has length $n$. Then \[s_\lambda(1,1,\dots,1)=\prod_{1\leq i<j\leq n}\frac{\lambda_i-\lambda_j+j-i}{j-i},\]
which is a polynomial in $\lambda_i-\lambda_j$. 
\end{lemma}

Since the set of RSSYT is a proper subset of all SSYT, we have that the number of RSSYT of rectangular shape $\lambda=\langle N^{k\beta}\rangle$ is { bounded by a }polynomial in $N$ of degree $k^2\beta^2$. This concludes the proof of the bound on $\mom_N(k,\beta)$ for integer values of $k$, $\beta$, and $N$.

%%%%%%%%%%%%%%%%%%%
%                                                     %
%                                                     %
%                       CA                         %
%                                                     %
%                                                     %
%%%%%%%%%%%%%%%%%%%
\section{Multiple Integrals}

In this section we prove \cref{thm:mom}.  We rely on a series of lemmas which we state in subsection~\ref{proofoutline} and then prove in subsection~\ref{statementofproof}.

\subsection{Proof outline}\label{proofoutline}
Recall from the introduction and subsection~\ref{knownresults} that theorem~\ref{thm:mom} is known for $k=1, 2$, so we will henceforth focus on integers $k>2$ (though the method we now develop can be adapted for the cases $k=1, 2$ as well). A key element of the proof is the following result (lemma 2.1) of Conrey et al.~\cite{cfkrs1}. 

\begin{lemma}\label{lemma:mci}
For $\alpha_j\in\mathbb{C}$,
\begin{align*}
\int_{U(N)}&\prod_{j=m+1}^n\det(I-Ae^{\alpha_j})\prod_{j=1}^m\det(I-A^*e^{-\alpha_j})dA\\
&=\frac{(-1)^{n(n-1)/2}}{(2\pi i)^nm!(n-m)!}\prod_{q=m+1}^{n}e^{N\alpha_q}\oint\cdots\oint\frac{e^{-N\sum_{l=m+1}^nz_l}\Delta(z_1,\dots,z_{n})^2dz_1\cdots dz_n}{\prod_{1\leq l\leq m<q\leq n}\left(1-e^{z_q-z_l}\right)\prod_{l=1}^n\prod_{q=1}^n(z_l-\alpha_q)},
\end{align*}
where the contours enclose the poles at $\alpha_1,\dots,\alpha_n$ and $\Delta(z_1,\dots,z_n)=\prod_{i<j}(z_j-z_i)$ is the Vandermonde determinant.
\end{lemma}

Before using \cref{lemma:mci}, we first define 
\begin{equation}\label{eq:pureca}
I_{k,\beta}(\theta_1,\dots,\theta_k)\coloneqq\mathbb{E}_{A\in U(N)}\left(\prod_{j=1}^k|P_N(A,\theta_j)|^{2\beta}\right),
\end{equation}
which captures the average over the unitary group and thus 
\begin{equation}\label{eq:momall}
\mom_N(k,\beta)=\frac{1}{(2\pi)^k}\int_0^{2\pi}\cdots\int_0^{2\pi}I_{k,\beta}(\theta_1,\dots,\theta_k)d\theta_1\cdots d\theta_k.
\end{equation}

Our focus now switches to understanding $I_{k,\beta}(\underline{\theta})$.  Following Keating et al.~\cite{krrr15}, we use \cref{lemma:mci} to expand the average over the CUE to a multiple contour integral.

\begin{align}
\label{eq:intfor}
I_{k,\beta}(\underline{\theta})&=\frac{(-1)^{k\beta}e^{-i\beta N\sum_{j=1}^k\theta_{j}}}{(2\pi i)^{2k\beta}((k\beta)!)^2}\oint\cdots\oint\frac{e^{-N(z_{k\beta+1}+\cdots+z_{2k\beta})}\Delta(z_1,\dots,z_{2k\beta})^2dz_1\cdots dz_{2k\beta}}{\prod_{m\leq k\beta<n}\left(1-e^{z_n-z_m}\right)\prod_{m=1}^{2k\beta}\prod_{n=1}^{k}(z_m+i\theta_n)^{2\beta}}.
\end{align}
We note that equations~(\ref{eq:momall}) and (\ref{eq:intfor}) define  $\mom_N(k,\beta)$ as an analytic function of $N$.

We deform each of the $2k\beta$ contours so that any one now consists of a sum of $k$ small circles surrounding each of the poles at $-i\theta_1,\dots,-i\theta_k$, given by $\Gamma_{-i\theta_l}$ for $l\in\{1,\dots,k\}$, and connecting straight lines whose contributions will cancel (just as in \cite{krrr15}, we follow the procedure outlined in ~\cite{keaodg08}).  This means that we will have a sum of $k^{2k\beta}$ multiple integrals,
\begin{equation}\label{eq:expansion}
I_{k,\beta}(\underline{\theta})=\frac{(-1)^{k\beta}e^{-i\beta N\sum_{j=1}^k\theta_{j}}}{(2\pi i)^{2k\beta}((k\beta)!)^2}\sum_{\varepsilon_j\in\{1,\dots,k\}}J_{k,\beta}(\underline{\theta};\varepsilon_1,\dots,\varepsilon_{2k\beta}),
\end{equation}
where 
\begin{equation}
J_{k,\beta}(\underline{\theta};\varepsilon_1,\dots,
\varepsilon_{2k\beta})=\int_{\Gamma_{-i\theta_{\varepsilon_1}}}\cdots\int_{\Gamma_{-i\theta_{\varepsilon_{2k\beta}}}}\frac{e^{-N(z_{k\beta+1}+\cdots+z_{2k\beta})}\Delta(z_1,\dots,z_{2k\beta})^2dz_1\cdots dz_{2k\beta}}{\prod_{m\leq k\beta<n}\left(1-e^{z_n-z_m}\right)\prod_{m=1}^{2k\beta}\prod_{n=1}^{k}(z_m+i\theta_n)^{2\beta}}
\end{equation}
is the multiple contour integral with $2k\beta$ contours each specialised around one of the $k$ poles determined by the vector $\underline{\varepsilon}=(\varepsilon_1,\dots,\varepsilon_{2k\beta})$.

In fact, many of the summands do not contribute to the sum due to the highly symmetric nature of the integrand.  The following lemma, which is a generalised version of lemma 4.11 in~\cite{krrr15}, determines exactly which summands make no contribution.  

\begin{lemma}\label{lemma:symmetric}
Let a choice of contours in \cref{eq:expansion} be denoted by $\underline{\varepsilon}=(\varepsilon_1,\dots,\varepsilon_{2k\beta})$ where $\varepsilon_j\in\{1,\dots,k\}$.  If any particular pole is overrepresented in $\underline{\varepsilon}$ (i.e. some pole $-i\theta^*$ features in at least $2\beta+1$ contours), then that summand is identically zero.  
\end{lemma}

Thus we have that
\begin{equation}
I_{k,\beta}(\underline{\theta})=\frac{(-1)^{k\beta}e^{-i\beta N\sum_{j=1}^k\theta_{j}}}{(2\pi i)^{2k\beta}((k\beta)!)^2}\sum_{l_1=0}^{2\beta}\cdots\sum_{l_{k-1}=0}^{2\beta}c_{\underline{l}}(k,\beta)J_{k,\beta;\underline{l}}(\underline{\theta}),
\end{equation}
where $J_{k,\beta;\underline{l}}(\underline{\theta})$ is the integral $J_{k,\beta}(\underline{\theta};\underline{\varepsilon})$ with contours given by 
\[\underline{\varepsilon}=(\overbrace{1,\dots,1}^{l_1},\overbrace{2,\dots,2}^{l_2},\dots,\overbrace{k-1,\dots,k-1}^{l_{k-1}},\overbrace{k,\dots,k}^{2\beta},\overbrace{k-1,\dots,k-1}^{2\beta-l_{k-1}},\dots,\overbrace{1,\dots,1}^{2\beta-l_1}),\]
and $c_{\underline{l}}(k,\beta)$ is a product of binomial coefficients capturing the symmetry exhibited by the integrand: 
\begin{align}
c_{\underline{l}}(k,\beta)&=\binom{k\beta}{l_1}\binom{k\beta-l_1}{l_2}\binom{k\beta-(l_1+l_2)}{l_3}\cdots\binom{k\beta-\sum_{m=1}^{k-2}l_m}{l_{k-1}}\nonumber\\
&\times\binom{k\beta}{2\beta-l_1}\binom{(k-2)\beta+l_1}{2\beta-l_2}\cdots\binom{k\beta-\sum_{m=1}^{k-2}(2\beta-l_m)}{2\beta-l_{k-1}}.\label{prodofbinom}
\end{align}
So $c_{\underline{l}}(k,\beta)$ counts the number of ways of picking $l_1$ of the first $k\beta$ contours and $2\beta-l_1$ of the second $k\beta$ contours to surround $-i\theta_1$, and then repeating on the remaining $k\beta-l_1$ contours in the first half and $(k-2)\beta+l_1$ contours in the second half, and so on. 

%Note also that the coefficient $c_{\underline{l}}(k,\beta)$ will not allow `overcrowding' of either half.  If for example we were to set $l_1=\dots=l_{k-1}=2\beta$ then we would be trying to fit $2(k-1)\beta$ labels on to $k\beta$ contours, which clearly cannot be done since we assume $k>2$.  However, $c_{\underline{l}}(k,\beta)$ contains the binomial coefficient 
%\[\binom{k\beta-\sum_{m=1}^{k-2}l_m}{l_{k-1}}=\binom{(4-k)\beta}{2\beta}=0\quad\text{for }k> 2,\]
%and so is zero for this choice of $\underline{l}$.  

Next we perform the change of variables, 
\[z_n=\frac{v_n}{N}-i\alpha_n,\]
where 
\begin{equation}\label{alphavector}
\alpha_n=\begin{cases}\theta_1&n\in\{1,\dots,l_1\}\cup\{2(k-1)\beta+1+l_1,\dots,2k\beta\}\\
\theta_2&n\in\{l_1+1,\dots,l_1+l_2\}\cup\{2(k-2)\beta+1+l_1+l_2,\dots,2(k-1)\beta+l_1\}\\
\vdots&\quad\vdots\\
\theta_{k-1}&n\in\{\sum_{m=1}^{k-2}l_m+1,\dots,\sum_{m=1}^{k-1}l_m\}\cup\{2\beta+1+\sum_{m=1}^{k-1}l_m,\dots,4\beta+\sum_{m=1}^{k-2}l_m\}\\
\theta_k&n\in\{\sum_{m=1}^{k-1}l_m+1,\dots,\sum_{m=1}^{k-1}l_m+2\beta\}.\end{cases}
\end{equation}
which shifts all the contours to be small circles surrounding the origin. Then up to terms of order $1/N$ smaller\footnote{{Henceforth, whenever we write `up to terms of order $1/N$' followed by a statement of the form $A(N)\sim B(N)$ we mean that $A(N)=B(N)(1+O(1/N))$ as $N\rightarrow\infty$.}}, we have that the integrand of $J_{k,\beta;\underline{l}}(\underline{\theta})$ is
\begin{align}
&\frac{e^{-\sum_{m=k\beta+1}^{2k\beta}v_m}e^{iN\sum_{m=k\beta+1}^{2k\beta}\alpha_m}\prod_{\substack{m<n\\\alpha_m\neq\alpha_n}}(i\alpha_m-i\alpha_n)^2\prod_{\substack{m<n\\\alpha_m=\alpha_n}}\left(\frac{v_n-v_m}{N}\right)^2\prod_{m=1}^{2k\beta}\frac{dv_m}{N}}{\prod_{\substack{m\leq k\beta<n\\\alpha_n\neq\alpha_m}}\left(1-e^{\frac{v_n-v_m}{N}}e^{i(\alpha_m-\alpha_n)}\right)\prod_{\substack{m\leq k\beta<n\\\alpha_m=\alpha_n}}\left(\frac{v_m-v_n}{N}\right)\prod_{m=1}^{2k\beta}\prod_{n=1}^k\left(\frac{v_m}{N}+i(\theta_n-\alpha_m)\right)^{2\beta}}\nonumber\\
&\quad\sim\frac{e^{iN\sum_{m=k\beta+1}^{2k\beta}\alpha_m}}{N^{2k\beta}}\frac{e^{-\sum_{m=k\beta+1}^{2k\beta}v_m}\prod_{\substack{m<n\\\alpha_m\neq\alpha_n}}(i\alpha_m-i\alpha_n)^2\prod_{\substack{m<n\\\alpha_m=\alpha_n}}\left(\frac{v_n-v_m}{N}\right)^2\prod_{m=1}^{2k\beta}\left(\frac{v_m}{N}\right)^{-2\beta}\prod_{m=1}^{2k\beta}{dv_m}}{\prod_{\substack{m\leq k\beta<n\\\alpha_n\neq\alpha_m}}\left(1-e^{\frac{v_n-v_m}{N}}e^{i(\alpha_m-\alpha_n)}\right)\prod_{\substack{m\leq k\beta<n\\\alpha_m=\alpha_n}}\left(\frac{v_m-v_n}{N}\right)\prod_{\substack{m< n\\\alpha_m\neq\alpha_n}}(i\alpha_m-i\alpha_n)^2}\\
&\quad=\frac{e^{iN\sum_{m=k\beta+1}^{2k\beta}\alpha_m}N^{4k\beta^2}}{N^{2k\beta}}\frac{e^{-\sum_{m=k\beta+1}^{2k\beta}v_m}\prod_{\substack{m<n\\\alpha_m=\alpha_n}}\left(\frac{v_n-v_m}{N}\right)^2\prod_{m=1}^{2k\beta}\frac{dv_m}{v_m^{2\beta}}}{\prod_{\substack{m\leq k\beta<n\\\alpha_n\neq\alpha_m}}\left(1-e^{\frac{v_n-v_m}{N}}e^{i(\alpha_m-\alpha_n)}\right)\prod_{\substack{m\leq k\beta<n\\\alpha_m=\alpha_n}}\left(\frac{v_m-v_n}{N}\right)}.\label{midintegrand}
\end{align}

To determine the power of $N$ coming from the terms originating from the Vandermonde determinant, we count the sizes of the following sets,
\begin{align}
\#\{(m,n): 1\leq m<n\leq 2k\beta\}=\binom{2k\beta}{2}&=k\beta(2k\beta-1)\label{sizeofsetA}\\
\#\{(m,n): 1\leq m<n\leq 2k\beta, \alpha_m\neq\alpha_n\}&=2\beta^2k(k-1)\label{sizeofsetB}\\
\#\{(m,n): 1\leq m<n\leq 2k\beta, \alpha_m=\alpha_n\}&=k\beta(2\beta-1)\label{sizeofsetC}.
\end{align}

One sees that, for example, (\ref{sizeofsetB}) results from the following calculation.  Firstly, we recall the structure of the vector \underline{$\alpha$},
\begin{equation}
(\overbrace{\theta_1,\dots,\theta_1}^{l_1},\overbrace{\theta_2,\dots,\theta_2}^{l_2},\dots,\overbrace{\theta_{k-1},\dots,\theta_{k-1}}^{l_{k-1}},\underbrace{\theta_k,\dots,\theta_k}_{2\beta},\overbrace{\theta_{k-1},\dots,\theta_{k-1}}^{2\beta-l_{k-1}},\dots,\overbrace{\theta_1,\dots,\theta_1}^{2\beta-l_1}).
\end{equation}
As we are looking for pairs $(m,n)$ such that $m<n$ and $\alpha_m\neq\alpha_n$, it is clear that any of the first $l_1$ choices of $\theta_1$ can be paired with any of the following $2k\beta-l_1$ options, except for the final $2\beta-l_1$ as these are also $\theta_1$.  Thus, the total number of pairs $(m,n)$ with $m\in\{1,\dots,l_1\}$, $m<n$, and $\theta_n\neq \theta_1$ is $l_1(2k\beta-l_1-(2\beta-l_1))$.  Continuing in this fashion we see that in general, for $m\in\{1,\dots,\sum_{j=1}^{k-1}l_j\}$, the number of such pairs is given by considering
\begin{equation}
\sum_{i=1}^{k-1}l_i\Big(2\beta+ \sum_{j=i+1}^{k-1}l_j+\sum_{\substack{j=1\\j\neq i}}^{k-1}(2\beta-l_j)\Big).
\end{equation}
Similarly, for $m\in\{l_1+\cdots+l_{k-1}+1,\dots,l_1+\cdots+l_{k-1}+2\beta\}$, the number of pairs satisfying the correct conditions is 
\begin{equation}
\sum_{i=1}^{k-1}2\beta(2\beta-l_i),
\end{equation}
and if $m\in\{l_1+\cdots+l_{k-1}+2\beta+1,\dots,2k\beta\}$ then we get
\begin{equation}
\sum_{i=1}^{k-1}(2\beta-l_i)\sum_{j=1}^{i-1}(2\beta-l_j).
\end{equation}
In total therefore, we have to evaluate

\begin{equation}\label{calcforsizeofset}
\sum_{i=1}^{k-1}l_i\Big(2\beta+ \sum_{j=i+1}^{k-1}l_j+\sum_{\substack{j=1\\j\neq i}}^{k-1}(2\beta-l_j)\Big)+\sum_{i=1}^{k-1}2\beta(2\beta-l_i)+\sum_{i=1}^{k-1}(2\beta-l_i)\sum_{j=1}^{i-1}(2\beta-l_j).
\end{equation}
By collecting like terms we see that (\ref{calcforsizeofset}) is equal to
\begin{align}
2\beta(k-1)&\sum_{i=1}^{k-1}l_i-\sum_{i=1}^{k-1}\sum_{j=1}^{i-1}l_il_j+2\beta(2\beta(k-1)-\sum_{i=1}^{k-1}l_i)+\sum_{i=1}^{k-1}\sum_{j=1}^{i-1}(4\beta^2-2\beta(l_i+l_j)+l_il_j)\nonumber\\
&=2\beta(k-1)\sum_{i=1}^{k-1}l_i+2\beta(2\beta(k-1)-\sum_{i=1}^{k-1}l_i)+\sum_{i=1}^{k-1}\sum_{j=1}^{i-1}(4\beta^2-2\beta(l_i+l_j))\\
%&=2\beta^2k(k-1)+2\beta\left((k-2)\sum_{i=1}^{k-1}l_i-\sum_{i=1}^{k-1}\sum_{j=1}^{i-1}(l_i+l_j)\right)\\
&=2\beta^2k(k-1)+2\beta\left((k-2)\sum_{i=1}^{k-1}l_i-\sum_{i=1}^{k-1}(i-1)l_i-\sum_{j=1}^{k-2}(k-1-j)l_j\right)\\
%&=2\beta^2k(k-1)+2\beta\left((k-2)\sum_{i=1}^{k-1}l_i-\sum_{i=1}^{k-1}(i-1+k-1-i)l_i\right)\\
%&=2\beta^2k(k-1)+2\beta\left((k-2)\sum_{i=1}^{k-1}l_i-(k-2)\sum_{i=1}^{k-1}l_i\right)\\
&=2\beta^2k(k-1)
\end{align}
as claimed in (\ref{sizeofsetB}). Then (\ref{sizeofsetC}) can be deduced immediately since it is the difference between (\ref{sizeofsetA}) and (\ref{sizeofsetB}).

To count the remaining power of $N$ that remains in the denominator of the integrand in \eqref{midintegrand}, we define 
\begin{align}
A_{k,\beta;\underline{l}}&\coloneqq\{(m,n):1\leq m\leq k\beta<n\leq 2k\beta, \alpha_m=\alpha_n\}\\
B_{k,\beta;\underline{l}}&\coloneqq\{(m,n):1\leq m\leq k\beta<n\leq 2k\beta, \alpha_m\neq\alpha_n\},
\end{align}
so $|A_{k,\beta;\underline{l}}|+|B_{k,\beta;\underline{l}}|=k^2\beta^2$. Hence
\begin{align}
\prod_{\substack{m<n\\\alpha_m=\alpha_n}}\left(\frac{v_n-v_m}{N}\right)^2&=\frac{1}{N^{2k\beta(2\beta-1)}}\prod_{\substack{m<n\\\alpha_m=\alpha_n}}\left({v_n-v_m}\right)^2\\
\prod_{\substack{m\leq k\beta<n\\\alpha_m=\alpha_n}}\left(\frac{v_m-v_n}{N}\right)&=\frac{1}{(-N)^{|A_{k,\beta;\underline{l}}|}}\prod_{\substack{m\leq k\beta<n\\\alpha_m=\alpha_n}}\left({v_n-v_m}\right).
\end{align}

Returning once more to the integrand of $J_{k,\beta;\underline{l}}(\underline{\theta})$, we have that up to terms of order $1/N$ smaller it is equal to
\begin{equation}
{e^{iN\sum_{m=k\beta+1}^{2k\beta}\alpha_m}(-N)^{|A_{k,\beta;\underline{l}}|}}\frac{e^{-\sum_{m=k\beta+1}^{2k\beta}v_m}\prod_{\substack{m<n\\\alpha_m=\alpha_n}}\left({v_n-v_m}\right)^2\prod_{m=1}^{2k\beta}\frac{dv_m}{v_m^{2\beta}}}{\prod_{\substack{m\leq k\beta<n\\\alpha_n\neq\alpha_m}}\left(1-e^{\frac{v_n-v_m}{N}}e^{i(\alpha_m-\alpha_n)}\right)\prod_{\substack{m\leq k\beta<n\\\alpha_m=\alpha_n}}\left({v_n-v_m}\right)}.
\end{equation}

If we set $l_k=k\beta-(l_1+\cdots+l_{k-1})$, then  { the calculation of the size of $A_{k,\beta;\underline{l}}$ follows the method outlined for $(\ref{sizeofsetB})$ and we have}
\[(-1)^{|A_{k,\beta;\underline{l}}|}=(-1)^{\sum_{j=1}^{k}l_j(2\beta-l_j)}=(-1)^{k\beta},\]
and so 
\begin{equation}
I_{k,\beta}(\underline{\theta})\sim \sum_{l_1,\dots,l_{k-1}=0}^{2\beta}\frac{c_{\underline{l}}(k,\beta)N^{|A_{k,\beta;\underline{l}}|}}{(2\pi i)^{2k\beta}((k\beta)!)^2}\int_{\Gamma_0}\cdots\int_{\Gamma_0}\frac{e^{-iN(\beta \sum_{j=1}^k\theta_{j}-\sum_{m=k\beta+1}^{2k\beta}\alpha_m)}f(\underline{v};\underline{l})\prod_{m=1}^{2k\beta}dv_m}{\prod_{(m,n)\in B_{k,\beta;\underline{l}}}\left(1-e^{\frac{v_n-v_m}{N}}e^{i(\alpha_m-\alpha_n)}\right)},\label{eq:Ikbeta}
\end{equation}
where we isolate the terms with no $\theta$ dependence and denote them by $f(\underline{v};\underline{l})$, so explicitly
\begin{equation}\label{notheta}
f(\underline{v};\underline{l})=\frac{e^{-\sum_{m=k\beta+1}^{2k\beta}v_m}\prod_{\substack{m<n\\\alpha_m=\alpha_n}}\left({v_n-v_m}\right)^2}{\prod_{\substack{m\leq k\beta<n\\\alpha_m=\alpha_n}}\left({v_n-v_m}\right)\prod_{m=1}^{2k\beta}v_m^{2\beta}}.
\end{equation}

We now focus on the denominator of the integrand in \cref{eq:Ikbeta}, which is the final term involving $N$. It will prove to be fruitful to use the properties of the $\alpha_j$s to remove the dependence on $(m,n)\in B_{k,\beta;\underline{l}}$. To this end, one can split the set $B_{k,\beta;\underline{l}}$ into $\binom{k}{2}$ disjoint subsets: 
\begin{equation}\label{bdisjointsets}
B_{k,\beta;\underline{l}}=\bigcup_{1\leq\sigma<\tau\leq k}S_{\sigma,\tau},
\end{equation}
where 
\begin{equation}\label{bpartitionA}
S_{\sigma,\tau}\coloneqq\{(m,n)\in B_{k,\beta;\underline{l}}:\alpha_m-\alpha_n=\pm(\theta_\tau-\theta_\sigma)\},
\end{equation}
and further partition $S_{\sigma,\tau}$ into two subsets $S^+_{\sigma,\tau}, S^-_{\sigma,\tau}$, where 
\begin{align}
S^+_{\sigma,\tau}&=\{(m,n)\in B_{k,\beta;\underline{l}}:\alpha_m-\alpha_n=\theta_\tau-\theta_\sigma\}\\
S^-_{\sigma,\tau}&=\{(m,n)\in B_{k,\beta;\underline{l}}:\alpha_m-\alpha_n=\theta_\sigma-\theta_\tau\}.
\end{align}

The goal of the next lemma is to use the structure of the vector $\underline{\alpha}$, to `decouple' the term $\exp(i(\alpha_m-\alpha_n))$ from the pair $(m,n)$. 

\begin{lemma}\label{lemma:geosums}

\begin{align*}
&\prod_{(m,n)\in B_{k,\beta;\underline{l}}}\left(1-\exp\left(\frac{v_n-v_m}{N}\right)\exp\left(i(\alpha_m-\alpha_n)\right)\right)^{-1}\\
&\quad=(-1)^{g(k,\beta;\underline{l})}\sum_{\substack{t_{\sigma,\tau}=-\infty\\\text{for }1\leq\sigma<\tau\leq k}}^\infty\exp\left(i\sum_{\gamma<\rho}(\theta_\rho-\theta_\gamma)(t_{\gamma,\rho}+|S^-_{\gamma,\rho}|)\right)\sum_{\substack{(x_{m,n})\\ \textbf{(}\textbf{$\star$}\textbf{)}}}\exp\left({\tfrac{1}{N}\sum x_{m,n}(v_n-v_m)^{\pm}}\right),
\end{align*}

where \[(v_n-v_m)^\pm=\begin{cases}v_n-v_m&\text{if }(m,n)\in S_{\sigma,\tau}^+\\v_m-v_n&\text{if }(m,n)\in S_{\sigma,\tau}^-,\end{cases}\]
the first sum appearing in the right hand size of the statement of the lemma should be read as 
\[\sum_{\substack{t_{\sigma,\tau}=-\infty\\\text{for }1\leq\sigma<\tau\leq k}}^\infty=\sum_{t_{1,2}=-\infty}^\infty\sum_{t_{1,3}=-\infty}^\infty\cdots\sum_{t_{1,k}=-\infty}^\infty\sum_{t_{2,3}=-\infty}^\infty\cdots\sum_{t_{k-1,k}=-\infty}^\infty\]  
and the second sum is over vectors (of weights) $\underline{x}=(x_{m,n})_{(m,n)\in B_{k,\beta;\underline{l}}}$ subject to constraints given by ($\star$), which are 
\begin{align}
\sum_{(m,n)\in S_{\sigma,\tau}}x_{m,n}&=t_{\sigma,\tau}+|S_{\sigma,\tau}^-|,\quad 1\leq \sigma<\tau\leq k\label{bconstraint}\\
x_{m,n}&\in \mathbb{Z}, \quad \forall (m,n)\in B_{k,\beta;\underline{l}}\\
H(-x_{m,n}\RE\{(v_n-v_m)^\pm\})&=1,\quad\forall (m,n)\in B_{k,\beta;\underline{l}},
\end{align}
where $\operatorname{H}(x)$ is the Heaviside step function (so $\operatorname{H}(x)=1$ if $x\geq0$ and $\operatorname{H}(x)=0$ if $x<0$).  Furthermore, the prefactor can be expressed as follows, 
\[(-1)^{g(k,\beta;\underline{l})}=(-1)^{\sum_{\sigma<\tau}|S_{\sigma,\tau}^-|}\prod_{\substack{(m,n)\in S_{\sigma,\tau}\\ 1\leq \sigma<\tau\leq k}} (-\operatorname{sign}(\RE\{(v_n-v_m)^\pm\})).\]
\end{lemma}

Using \cref{lemma:geosums} with \cref{eq:Ikbeta} we have
\begin{align}
I_{k,\beta}(\underline{\theta})&\sim \sum_{l_1,\dots,l_{k-1}=0}^{2\beta}\frac{\tilde{c}_{\underline{l}}(k,\beta)N^{|A_{k,\beta;\underline{l}}|}}{(2\pi i)^{2k\beta}((k\beta)!)^2}\int_{\Gamma_0}\cdots\int_{\Gamma_0}f(\underline{v};\underline{l})\exp\left(iN\Big(\sum_{m=k\beta+1}^{2k\beta}\alpha_m-\beta \sum_{j=1}^k\theta_{j}\Big)\right)\nonumber\\
&\quad\times\sum_{\substack{t_{\sigma,\tau}=-\infty\\\text{for }1\leq\sigma<\tau\leq k}}^\infty\exp\left(i\sum_{{1\leq}\gamma<\rho{\leq k}}(\theta_\rho-\theta_\gamma)(t_{\gamma,\rho}+|S^-_{\gamma,\rho}|)\right)\sum_{\substack{(x_{m,n})\\ \textbf{(}\textbf{$\star$}\textbf{)}}}\exp\left({\tfrac{1}{N}\sum x_{m,n}(v_n-v_m)^{\pm}}\right)\prod_{m=1}^{2k\beta}dv_m,\label{eq:IkbetaGEO}
\end{align}
where 
\[\tilde{c}_{\underline{l}}(k,\beta)=(-1)^{g(k,\beta;\underline{l})}c_{\underline{l}}(k,\beta),\]
and the function $g(k,\beta;\underline{l})$ is as described in the statement of \cref{lemma:geosums}. Now we relate $I_{k,\beta}(\underline{\theta})$ back to $\mom_N(k,\beta)$ using \cref{eq:momall}, 
and deduce the following lemma. 

\begin{lemma}\label{momlemma1}
\begin{align*}
\mom_N(k,\beta)&\sim\sum_{l_1,\dots,l_{k-1}=0}^{2\beta}\frac{\tilde{c}_{\underline{l}}(k,\beta)N^{|A_{k,\beta;\underline{l}}|}}{(2\pi i)^{2k\beta}((k\beta)!)^2}\int_{\Gamma_0}\cdots\int_{\Gamma_0}f(\underline{v};\underline{l})\hspace{-0.3cm}\sum_{\substack{t_{\sigma,\tau}=-\infty\\\text{for }1\leq\sigma<\tau\leq k}}^\infty\sum_{\substack{(x_{m,n})\\ \textbf{(}\textbf{$\star$}\textbf{)}}}\exp\left({\tfrac{1}{N}\sum x_{m,n}(v_n-v_m)^{\pm}}\right)\\
&\times\prod_{j=1}^{k-1}\delta_{N(l_j-\beta)+\sum_{\rho=j+1}^k(t_{j,\rho}+|S^-_{j,\rho}|)-\sum_{\gamma=1}^{j-1}(t_{\gamma,j}+|S^-_{\gamma,j}|)}\prod_{m=1}^{2k\beta}dv_m\nonumber.
\end{align*}
\end{lemma}

In order to take the asymptotic analysis further we require the next lemma. 

\begin{lemma}\label{momlemma2}
\begin{align*}
&\sum_{\substack{t_{\sigma,\tau}=-\infty\\\text{for }1\leq\sigma<\tau\leq k}}^\infty\sum_{\substack{(x_{m,n})\\ \textbf{(}\textbf{$\star$}\textbf{)}}}\exp\left({\tfrac{1}{N}\sum x_{m,n}(v_n-v_m)^{\pm}}\right)\prod_{j=1}^{k-1}\delta_{N(l_j-\beta)+\sum_{\rho=j+1}^k(t_{j,\rho}+|S^-_{j,\rho}|)-\sum_{\gamma=1}^{j-1}(t_{\gamma,j}+|S^-_{\gamma,j}|)}\\%\label{eq:Nasympt1}\\
&\sim N^{|B_{k,\beta;\underline{l}}|-k+1}\kappa_k\left((k-1)\beta-\sum_{j=1}^{k-1}l_j\right)^{|B_{k,\beta;\underline{l}}|-\binom{k}{2}}\Psi_{k,\beta;\underline{l}}(((k-1)\beta-\sum_{j=1}^{k-1}l_j)\underline{v}),\nonumber
\end{align*}
where $\kappa_k$ is a constant depending on $k$,  \[\Psi_{k,\beta;\underline{l}}(\underline{v}) =\underset{\substack{\underline{y}=(y_{m,n})_{(m,n)\in B_{k,\beta;\underline{l}}}\\\textbf{(}\tilde{\textbf{$\ddagger$}}\textbf{)}}}{\int\cdots\int}\exp\left({\sum y_{m,n}(v_n-v_m)^{\pm}}\right)\prod dy_{m,n},\]
and $\textbf{(}\tilde{\textbf{$\ddagger$}}\textbf{)}$ denotes normalised constraints related to those previously denoted \textbf{(}\textbf{$\star$}\textbf{)}, see proof for more details.
\end{lemma}

Using lemma~\ref{momlemma1} and lemma~\ref{momlemma2} we can prove the following which establishes out the power of $N$ we seek. 

\begin{lemma}\label{momlemma3}
\[\mom_N(k,\beta)\sim \gamma_{k,\beta}N^{k^2\beta^2-k+1}\]
where 
\begin{equation*}\label{eq:gamma}
\gamma_{k,\beta}= \sum_{l_1,\dots,l_{k-1}=0}^{2\beta}{c}_{k,\beta;\underline{l}}((k-1)\beta-\sum_{j=1}^{k-1}l_j)^{|B_{k,\beta;\underline{l}}|-\binom{k}{2}}P_{k,\beta}(l_1,\dots,l_{k-1}),
\end{equation*}
${c}_{k,\beta;\underline{l}}$ is some constant depending on $k,\beta, \underline{l}$, and 
\begin{align*}
P_{k,\beta}(l_1,\dots,l_{k-1})&=\frac{(-1)^{g(k,\beta;\underline{l})}}{(2\pi i)^{2k\beta}((k\beta)!)^2}\\
&\times\int_{\Gamma_0}\cdots\int_{\Gamma_0}\frac{e^{-\sum_{m=k\beta+1}^{2k\beta}v_m}\prod_{\substack{m<n\\\alpha_m=\alpha_n}}\left({v_n-v_m}\right)^2}{\prod_{\substack{m\leq k\beta<n\\\alpha_m=\alpha_n}}\left({v_n-v_m}\right)\prod_{m=1}^{2k\beta}v_m^{2\beta}}\Psi_{k,\beta;\underline{l}}(((k-1)\beta-\sum_{j=1}^{k-1}l_j)\underline{v})\prod_{m=1}^{2k\beta}dv_m,\nonumber
\end{align*}
with $\Psi_{k,\beta;\underline{l}}(\underline{v})$ as defined in \cref{momlemma2}, and $g(k,\beta;\underline{l})$ given by \eqref{overallphase}.
\end{lemma}

The last step is to prove that we do indeed have the correct asymptotic, which is achieved through the final lemma.

\begin{lemma}\label{momlemma4}
For $k,\beta\in\mathbb{N}$, $\gamma_{k,\beta}\neq 0$ where $\gamma_{k,\beta}$ is as defined in \cref{momlemma3}. 
\end{lemma}

Hence combining lemma~\ref{momlemma3} and lemma \ref{momlemma4} gives us theorem~\ref{thm:mom}, 
\[\mom_N(k,\beta)=\gamma_{k,\beta}N^{k^2\beta^2-k+1}+O(N^{k^2\beta^2-k}).\]

\subsection{Proof details}\label{statementofproof}

%%%%%%%%%%%%%
%
%  Symm proof
%
%%%%%%%%%%%%%%

\begin{proof}[Proof of \cref{lemma:symmetric}]

We recall the statement of the lemma.
\begin{lemmum}
Let a choice of contours in 
\[J_{k,\beta}(\underline{\theta};\varepsilon_1,\dots,
\varepsilon_{2k\beta})=\int_{\Gamma_{-i\theta_{\varepsilon_1}}}\cdots\int_{\Gamma_{-i\theta_{\varepsilon_{2k\beta}}}}\frac{e^{-N(z_{k\beta+1}+\cdots+z_{2k\beta})}\Delta(z_1,\dots,z_{2k\beta})^2dz_1\cdots dz_{2k\beta}}{\prod_{m\leq k\beta<n}\left(1-e^{z_n-z_m}\right)\prod_{m=1}^{2k\beta}\prod_{n=1}^{k}(z_m+i\theta_n)^{2\beta}}
\]
be denoted by $\underline{\varepsilon}=(\varepsilon_1,\dots,\varepsilon_{2k\beta})$ where $\varepsilon_j\in\{1,\dots,k\}$.  If any one of the $k$ poles is overrepresented in $\underline{\varepsilon}$ (i.e. some pole $-i\theta^*$, $\theta^*\in\{\theta_1,\dots,\theta_k\}$, features in at least $2\beta+1$ contours), then for that choice of $\underline{\varepsilon}$, $J_{k,\beta}(\underline{\theta};\underline{\varepsilon})$ is identically zero.  
\end{lemmum}
The proof closely follows the proof of lemma 4.11 in~\cite{krrr15}. We show the case where $-i\theta_1$ is `overrepresented', and without loss of generality assume that the choice of contour is given by
\[{\underline{\varepsilon}}^*=(\overbrace{1,\dots,1}^{2\beta+1},\overbrace{2,\dots,2}^{2\beta-1},\overbrace{3,\dots,3}^{2\beta},\dots,\overbrace{k,\dots,k}^{2\beta}).\]
The other cases follow similarly.  

To $J_{k,\beta}(\underline{\theta},\underline{\varepsilon}^*)$ apply the change of variable $z_j\mapsto z_j-i\theta_1$ and consider the function, 
\[G(z_1,\dots,z_{2\beta+1})=\frac{e^{-N(z_{k\beta+1}+\cdots+z_{2k\beta})}\Delta(z_1,\dots,z_{2k\beta})}{\prod_{m\leq k\beta<n}\left(1-e^{z_n-z_m}\right)\prod_{m=1}^{2k\beta}\prod_{n=2}^{k}(z_m+i(\theta_n-\theta_1))^{2\beta}\prod_{m=2\beta+2}^{2k\beta}z_m^{2\beta}},\]
which is analytic around zero.  The integrand of $J_{k,\beta}(\underline{\theta};{\underline{\varepsilon}}^*)$ is
\begin{equation*}\label{proofexp}
e^{iNk\beta\theta_1}\frac{G(z_1,\dots,z_{2\beta+1})\Delta(z_1,\dots,z_{2k\beta})dz_1\cdots dz_{2k\beta}}{\prod_{m=1}^{2\beta+1}z_m^{2\beta}}.
\end{equation*}

We appeal to the residue theorem to compute $J_{k,\beta}(\underline{\theta};{\underline{\varepsilon}}^*)$, and the proof follows if we can show that the coefficient of $\prod_{m=1}^{2\beta+1}z_m^{2\beta-1}$ in $G(z_1,\dots,z_{2\beta+1})\Delta(z_1,\dots,z_{2k\beta})$ is zero.  Since $G(z_1,\dots,z_{2\beta+1})$ is analytic around zero, we focus on the Vandermonde determinant and use the following expansion,
\[\Delta(z_1,\dots,z_{2k\beta})=\sum_{\sigma\in S_{2k\beta}}\Sgn(\sigma)\prod_{m=1}^{2k\beta}z_m^{\sigma(m)-1}.\]
 Thus, we are searching for terms in this expansion of the form $\prod_{m=1}^{2\beta+1}z_m^{\sigma(m)-1}$ with $\sigma(m)-1\leq 2\beta-1$ for $m=1,\dots,2\beta+1$.  However, there is no term of this form as $\sigma$ is a permutation on the set $\{1,\dots,2k\beta\}$, so for at least one $m\in\{1,\dots,2\beta+1\}$, $\sigma(m)\geq 2\beta+1$. By the residue theorem we conclude that $J_{k,\beta}(\underline{\theta};{\underline{\varepsilon}}^*)$ is zero. 
\end{proof}

%%%%%%%%%%%%%
%
%  Geo series proof
%
%%%%%%%%%%%%%%

\begin{proof}[Proof of \cref{lemma:geosums}]
We recall the statement of the lemma.
\begin{lemmum}

\begin{align*}
&\prod_{(m,n)\in B_{k,\beta;\underline{l}}}\left(1-\exp\left(\frac{v_n-v_m}{N}\right)\exp\left(i(\alpha_m-\alpha_n)\right)\right)^{-1}\\
&\quad=(-1)^{g(k,\beta;\underline{l})}\sum_{\substack{t_{\sigma,\tau}=-\infty\\\text{for }1\leq\sigma<\tau\leq k}}^\infty\exp\left(i\sum_{\gamma<\rho}(\theta_\rho-\theta_\gamma)(t_{\gamma,\rho}+|S^-_{\gamma,\rho}|)\right)\sum_{\substack{(x_{m,n})\\ \textbf{(}\textbf{$\star$}\textbf{)}}}\exp\left({\tfrac{1}{N}\sum x_{m,n}(v_n-v_m)^{\pm}}\right),
\end{align*}

where \[(v_n-v_m)^\pm=\begin{cases}v_n-v_m&\text{if }(m,n)\in S_{\sigma,\tau}^+\\v_m-v_n&\text{if }(m,n)\in S_{\sigma,\tau}^-,\end{cases}\]
the first sum appearing in the right hand size of the statement of the lemma should be read as 
\[\sum_{\substack{t_{\sigma,\tau}=-\infty\\\text{for }1\leq\sigma<\tau\leq k}}^\infty=\sum_{t_{1,2}=-\infty}^\infty\sum_{t_{1,3}=-\infty}^\infty\cdots\sum_{t_{1,k}=-\infty}^\infty\sum_{t_{2,3}=-\infty}^\infty\cdots\sum_{t_{k-1,k}=-\infty}^\infty\] 
and the second sum is over vectors (of weights) $\underline{x}=(x_{m,n})_{(m,n)\in B_{k,\beta;\underline{l}}}$ subject to constraints given by ($\star$), which are 
\begin{align*}
\sum_{(m,n)\in S_{\sigma,\tau}}x_{m,n}&=t_{\sigma,\tau}+|S_{\sigma,\tau}^-|,\quad 1\leq \sigma<\tau\leq k\\
x_{m,n}&\in \mathbb{Z}, \quad \forall (m,n)\in B_{k,\beta;\underline{l}}\\
H(-x_{m,n}\RE\{(v_n-v_m)^\pm\})&=1,\quad\forall (m,n)\in B_{k,\beta;\underline{l}},
\end{align*}
where $\operatorname{H}(x)$ is the Heaviside step function (so $\operatorname{H}(x)=1$ if $x\geq0$ and $\operatorname{H}(x)=0$ if $x<0$).  Furthermore, the prefactor can be expressed as follows, 
\[(-1)^{g(k,\beta;\underline{l})}=(-1)^{\sum_{\sigma<\tau}|S_{\sigma,\tau}^-|}\prod_{\substack{(m,n)\in S_{\sigma,\tau}\\ 1\leq \sigma<\tau\leq k}} (-\operatorname{sign}(\RE\{(v_n-v_m)^\pm\})).\]
    
\end{lemmum}

{ Firstly, recall the definition of the sets $B_{k,\beta;\underline{l}}, S_{\sigma,\tau}^+$, and $S_{\sigma,\tau}^-$,
\begin{align*}
B_{k,\beta;\underline{l}}&\coloneqq\{(m,n):1\leq m\leq k\beta<n\leq 2k\beta, \alpha_m\neq\alpha_n\},\\
S^+_{\sigma,\tau}&\coloneqq\{(m,n)\in B_{k,\beta;\underline{l}}:\alpha_m-\alpha_n=\theta_\tau-\theta_\sigma\},\\
S^-_{\sigma,\tau}&\coloneqq\{(m,n)\in B_{k,\beta;\underline{l}}:\alpha_m-\alpha_n=\theta_\sigma-\theta_\tau\}.
\end{align*}}

We use the partition of $B_{k,\beta;\underline{l}}$ by the sets $S_{\sigma,\tau}^+, S_{\sigma,\tau}^-$ (although not emphasised in the notation, these sets also depend on $k, \beta$, and $l_1,\dots,l_{k-1}$) to break up the product appearing on the left hand side of the statement of the lemma as follows,
\begin{align}
\prod_{(m,n)\in B_{k,\beta;\underline{l}}}&\left(1-e^{\frac{v_n-v_m}{N}}e^{i(\alpha_m-\alpha_n)}\right)^{-1}\nonumber\\
&=\prod_{1\leq \sigma<\tau\leq k}\prod_{(m,n)\in S^+_{\sigma,\tau}}\left(1-e^{\frac{v_n-v_m}{N}}e^{i(\theta_\tau-\theta_\sigma)}\right)^{-1}\prod_{(p,q)\in S^-_{\sigma,\tau}}\left(1-e^{\frac{v_q-v_p}{N}}e^{i(\theta_\sigma-\theta_\tau)}\right)^{-1}\\
%&\quad=\prod_{1\leq \sigma<\tau\leq k}e^{i(\theta_\tau-\theta_\sigma)|S^-_{\sigma,\tau}|}\prod_{(m,n)\in S^+_{\sigma,\tau}}\left(1-e^{\frac{v_n-v_m}{N}}e^{i(\theta_\tau-\theta_\sigma)}\right)^{-1}\prod_{(p,q)\in S^-_{\sigma,\tau}}\left(e^{i(\theta_\tau-\theta_\sigma)}-e^{\frac{v_q-v_p}{N}}\right)^{-1}\\
&=\prod_{1\leq \sigma<\tau\leq k}(-1)^{|S^-_{\sigma,\tau}|}e^{i(\theta_\tau-\theta_\sigma)|S^-_{\sigma,\tau}|}\prod_{(p,q)\in S^-_{\sigma,\tau}}e^{\frac{v_p-v_q}{N}}\prod_{(m,n)\in S_{\sigma,\tau}}\left(1-e^{\frac{1}{N}(v_n-v_m)^\pm}e^{i(\theta_\tau-\theta_\sigma)}\right)^{-1}\label{geoterm}
\end{align}
where 
\begin{align*}
(v_n-v_m)^\pm&=\begin{cases}v_n-v_m&\text{for }(m,n)\in S_{\sigma,\tau}^+\\v_m-v_n&\text{for }(m,n)\in S_{\sigma,\tau}^-.\end{cases}\\
%\intertext{or equivalently,}
%(v_n-v_m)^\pm&=\begin{cases}v_n-v_m&\text{if }m\leq k\beta<n\text{ and }\varepsilon_m>\varepsilon_n\\
%v_m-v_n&\text{if }m\leq k\beta<n\text{ and }\varepsilon_m<\varepsilon_n,\end{cases}
\end{align*}
%{ which is equivalent to asking if $(m,n)\in S_{\sigma,\tau}^+$ or $S_{\sigma,\tau}^-$ respectively. }Recall that \[\underline{\varepsilon}=(\overbrace{1,\dots,1}^{l_1},\overbrace{2,\dots,2}^{l_2},\dots,\overbrace{k-1,\dots,k-1}^{l_{k-1}},\overbrace{k,\dots,k}^{2\beta},\overbrace{k-1,\dots,k-1}^{2\beta-l_{k-1}},\dots,\overbrace{1,\dots,1}^{2\beta-l_1}).\]

For a fixed choice $\sigma,\tau$, and a fixed pair $(m,n)$, we use the following expansion
\begin{align}
(1-&\exp(\tfrac{1}{N}(v_n-v_m)^\pm)\exp(i(\theta_\tau-\theta_\sigma)))^{-1}\nonumber\\
 &= -\operatorname{sign}(\RE\{(v_n-v_m)^\pm\})\sum_{t=-\infty}^\infty \exp(\tfrac{1}{N}(v_n-v_m)^{\pm}t)\exp(i(\theta_\tau-\theta_\sigma)t)\operatorname{H}(-t\RE\{(v_n-v_m)^\pm\}),\label{geoexpansion}
\end{align}
where $\operatorname{H}(x)$ is the Heaviside step function (so $\operatorname{H}(x)=1$ if $x\geq0$ and $\operatorname{H}(x)=0$ if $x<0$).  Note that this series converges because if $\RE\{(v_n-v_m)^\pm\}>0$, only negative $t$ terms contribute, and otherwise if $\RE\{(v_n-v_m)^\pm\}<0$, only non-negative $t$ terms survive.

%The above equality holds since when $\gamma_{m,n}$ is strictly negative, one can view the left hand side of \eqref{geoexpansion} as the evaluation of the following geometric series,
%
%\begin{equation}\label{expansiona}
%\left(1-\exp(\tfrac{1}{N}(v_n-v_m)^\pm)\exp(i(\theta_\tau-\theta_\sigma))\right)^{-1}  = \sum_{s=0}^\infty \exp(\tfrac{1}{N}(v_n-v_m)^{\pm}s)\exp(i(\theta_\tau-\theta_\sigma)s). 
%\end{equation}
%
% With $\gamma_{m,n}$ in this range, the Heaviside function in \eqref{geoexpansion} kills all negative values of $t$, and the resulting expression is precisely the right hand side of \eqref{expansiona}.  Otherwise, if $\gamma_{m,n}$ is strictly positive (by assumption it cannot be zero), then the obvious manipulation of the left hand side of \eqref{geoexpansion} is
%
%\begin{align}
%\frac{1}{1-\exp(\tfrac{1}{N}(v_n-v_m)^\pm)\exp(i(\theta_\tau-\theta_\sigma))}&=-\frac{\exp(-\tfrac{1}{N}(v_n-v_m)^\pm)\exp(-i(\theta_\tau-\theta_\sigma))}{1-\exp(-\tfrac{1}{N}(v_n-v_m)^\pm)\exp(-i(\theta_\tau-\theta_\sigma))}\\
%&=-\sum_{s=1}^{\infty}\exp(-\tfrac{1}{N}(v_n-v_m)^\pm s)\exp(-i(\theta_\tau-\theta_\sigma)s).\label{expansionb}
%\end{align}
%When $\gamma_{m,n}$ is strictly positive, then the only range of $t$ which survives in the right hand side of \eqref{geoexpansion} once more matches the summation range of \eqref{expansionb}.  
%
%
Now, incorporating \eqref{geoexpansion} into the final product of \eqref{geoterm}, we have

\begin{align}
&\prod_{(m,n)\in B_{k,\beta;\underline{l}}}\left(1-\exp\left({\tfrac{1}{N}(v_n-v_m)}\right)\exp\left({i(\alpha_m-\alpha_n)}\right)\right)^{-1}\nonumber\\
&=\prod_{1\leq \sigma<\tau\leq k}(-1)^{|S^-_{\sigma,\tau}|}\exp(i(\theta_\tau-\theta_\sigma)|S^-_{\sigma,\tau}|)\prod_{(p,q)\in S^-_{\sigma,\tau}}\exp(\frac{v_p-v_q}{N})\\
&\times\prod_{(m,n)\in S_{\sigma,\tau}} \left(-\operatorname{sign}(\RE\{(v_n-v_m)^\pm\})\sum_{t=-\infty}^\infty \exp(\tfrac{1}{N}(v_n-v_m)^{\pm}t)\exp(i(\theta_\tau-\theta_\sigma)t)\operatorname{H}(-t\RE\{(v_n-v_m)^\pm\})\right)\nonumber\\
%&=(-1)^{\sum_{\sigma<\tau}|S_{\sigma,\tau}^+|} \prod_{1\leq\sigma<\tau\leq k}\exp(i(\theta_\tau-\theta_\sigma)|S^-_{\sigma,\tau}|)\exp(\sum_{(p,q)\in S_{\sigma,\tau}^-}\frac{v_p-v_q}{N})\prod_{(m,n)\in S_{\sigma,\tau}}\operatorname{sign}(\RE\{(v_n-v_m)^\pm\})\\
%&\times\sum_{t=-\infty}^\infty \exp(\tfrac{1}{N}(v_n-v_m)^{\pm}t)\exp(i(\theta_\tau-\theta_\sigma)t)\operatorname{H}(-t\RE\{(v_n-v_m)^\pm\})\nonumber\\
%&=(-1)^{\sum_{\sigma<\tau}|S_{\sigma,\tau}^+|} \prod_{1\leq\sigma<\tau\leq k}\left(\prod_{(m,n)\in S_{\sigma,\tau}}\operatorname{sign}(\RE\{(v_n-v_m)^\pm\})\right)\exp(\sum_{(p,q)\in S_{\sigma,\tau}^-}\frac{v_p-v_q}{N})\\
%&\times\sum_{t_{\sigma,\tau}=-\infty}^\infty \exp(i(\theta_\tau-\theta_\sigma)(t_{\sigma,\tau}+|S_{\sigma,\tau}^-|))\hspace{-2.5em}\sum_{\substack{\underline{x}=(x_{m,n})\\\sum_{(m,n)\in S_{\sigma,\tau}} x_{m,n}=t_{\sigma,\tau}}} \hspace{-1em}\prod_{(m,n)\in S_{\sigma,\tau}}\exp(\tfrac{1}{N}(v_n-v_m)^{\pm}x_{m,n})\operatorname{H}(-x_{m,n}\RE\{(v_n-v_m)^\pm\}),\nonumber\\
%&=(-1)^{\sum_{\sigma<\tau}|S_{\sigma,\tau}^-|} \prod_{1\leq\sigma<\tau\leq k}\left(\prod_{(m,n)\in S_{\sigma,\tau}}(-\operatorname{sign}(\RE\{(v_n-v_m)^\pm\}))\right)\\
%&\times\sum_{t_{\sigma,\tau}=-\infty}^\infty \exp(i(\theta_\tau-\theta_\sigma)(t_{\sigma,\tau}+|S_{\sigma,\tau}^-|))\hspace{-4em}\sum_{\substack{\underline{x}=(x_{m,n})\\\sum_{(m,n)\in S_{\sigma,\tau}} x_{m,n}=t_{\sigma,\tau}+|S_{\sigma,\tau}^-|}} \hspace{-2em}\prod_{(m,n)\in S_{\sigma,\tau}}\exp(\tfrac{1}{N}(v_n-v_m)^{\pm}x_{m,n})\operatorname{H}(-x_{m,n}\RE\{(v_n-v_m)^\pm\})\nonumber\\
&=(-1)^{g(k,\beta;\underline{l})}\prod_{1\leq\sigma<\tau\leq k}\sum_{t_{\sigma,\tau}=-\infty}^\infty \exp(i(\theta_\tau-\theta_\sigma)(t_{\sigma,\tau}+|S_{\sigma,\tau}^-|))\\
&\times\sum_{\substack{\underline{x}=(x_{m,n})\\\sum_{(m,n)\in S_{\sigma,\tau}} x_{m,n}=t_{\sigma,\tau}+|S_{\sigma,\tau}^-|}} \prod_{(m,n)\in S_{\sigma,\tau}}\exp(\tfrac{1}{N}(v_n-v_m)^{\pm}x_{m,n})\operatorname{H}(-x_{m,n}\RE\{(v_n-v_m)^\pm\})\nonumber\\
&=(-1)^{g(k,\beta;\underline{l})}\sum_{\substack{t_{\sigma,\tau}=-\infty\\\text{for }1\leq\sigma<\tau\leq k}}^\infty\exp\left(i\sum_{\gamma<\rho}(\theta_\rho-\theta_\gamma)(t_{\gamma,\rho}+|S^-_{\gamma,\rho}|)\right)\sum_{\substack{(x_{m,n})\\ \textbf{(}\textbf{$\star$}\textbf{)}}}\exp\left({\tfrac{1}{N}\sum x_{m,n}(v_n-v_m)^{\pm}}\right).\label{multigeosum}
\end{align}

The overall sign in (\ref{multigeosum}) is 
\begin{equation}\label{overallphase}
(-1)^{g(k,\beta;\underline{l})}=(-1)^{\sum_{\sigma<\tau}|S_{\sigma,\tau}^-|}\prod_{\substack{(m,n)\in S_{\sigma,\tau}\\ 1\leq \sigma<\tau\leq k}} (-\operatorname{sign}(\RE\{(v_n-v_m)^\pm\})),
\end{equation}
the first sum should be read as 
\[\sum_{\substack{t_{\sigma,\tau}=-\infty\\\text{for }1\leq\sigma<\tau\leq k}}^\infty=\sum_{t_{1,2}=-\infty}^\infty\sum_{t_{1,3}=-\infty}^\infty\cdots\sum_{t_{1,k}=-\infty}^\infty\sum_{t_{2,3}=-\infty}^\infty\cdots\sum_{t_{k-1,k}=-\infty}^\infty,\] 
and the second sum is now over the `full' vectors $\underline{x}=(x_{m,n})_{(m,n)\in B_{k,\beta;\underline{l}}}$ whose elements $x_{m,n}$ are integers subject to constraints given by $(\star)$, which are
\begin{align}
\sum_{(m,n)\in S_{\sigma,\tau}}x_{m,n}&=t_{\sigma,\tau}+|S_{\sigma,\tau}^-|,\quad 1\leq \sigma<\tau\leq k\label{subseqconstraint1}\\
x_{m,n}&\in \mathbb{Z}, \quad \forall (m,n)\in B_{k,\beta;\underline{l}}\\
H(-x_{m,n}\RE\{(v_n-v_m)^\pm\})&=1,\quad\forall (m,n)\in B_{k,\beta;\underline{l}}.\label{subseqconstraint2}
\end{align}

This means that the vector $\underline{x}$ should be thought of as being made up of concatenated subsequences $(x_{m,n})_{(m,n)\in S_{\sigma,\tau}}$ for each $1\leq \sigma<\tau\leq k$, and each subsequence must satisfy the constraints \eqref{subseqconstraint1}-\eqref{subseqconstraint2}. This completes the proof of \cref{lemma:geosums}.

%Finally, setting $l_k$ to be the number of times $\theta_k$ occurs in the first $k\beta$ contours, so $l_k=k\beta-\sum_{m=1}^{k-1}l_m$, 
%
%\[\sum_{\sigma<\tau}|S_{\sigma,\tau}^-|=\sum_{\sigma=1}^{k-1}\sum_{\tau=\sigma+1}^k|S_{\sigma,\tau}^-|.\]
%We then recall the definition of $S_{\sigma,\tau}^-$ for a fixed choice of $\sigma, \tau,$ in order to determine its size,
%\[S_{\sigma,\tau}^-\coloneqq \{(m,n): 1\leq m\leq k\beta<n\leq 2k\beta, 0\neq\alpha_m-\alpha_n=\theta_\sigma-\theta_\tau\}.\]
%{Further, recall that the vector $\underline{\alpha}$ gives, for a fixed choice of $l_1,\dots,l_{k-1}$, the exact choice of contours, and the specific relation to the $\theta_j$ is}
%\[\underline{\alpha}=(\overbrace{\theta_1,\dots,\theta_1}^{l_1},\overbrace{\theta_2,\dots,\theta_2}^{l_2},\dots,\overbrace{\theta_{k-1},\dots,\theta_{k-1}}^{l_{k-1}},\overbrace{\theta_k,\dots,\theta_k}^{2\beta},\overbrace{\theta_{k-1},\dots,\theta_{k-1}}^{2\beta-l_{k-1}},\dots,\overbrace{\theta_1,\dots,\theta_1}^{2\beta-l_1}).\]
%Thus, for fixed $\sigma<\tau$, there are $l_\sigma$ choices for $m$ so that $\alpha_m=\theta_\sigma$, and $2\beta-l_{\tau}$ for $n$ such that $\alpha_n=\theta_\tau$ with $m,n$ in the required range. Hence $|S_{\sigma,\tau}^-|=l_\sigma(2\beta-l_\tau)$ and
%\[\sum_{\sigma=1}^{k-1}\sum_{\tau=\sigma+1}^k|S_{\sigma,\tau}^-|=\sum_{\sigma=1}^{k-1}\sum_{\tau=\sigma+1}^k l_\sigma(2\beta-l_\tau).\]
%This means that 
%\begin{equation*}
%(-1)^{\sum_{\sigma<\tau}|S_{\sigma,\tau}^-|}=(-1)^{\sum_{\sigma=1}^{k-1}\sum_{\tau=\sigma+1}^kl_\sigma l_\tau}
%\end{equation*}
%as was to be proved.

\end{proof}

%%%%%%%%%%%%%
%
%  MOM:First lemma
%
%%%%%%%%%%%%%%

\begin{proof}[Proof of \cref{momlemma1}]
Firstly, recall the statement of the lemma.  
\begin{lemmum}
\begin{align*}
\mom_N(k,\beta)&\sim\sum_{l_1,\dots,l_{k-1}=0}^{2\beta}\frac{\tilde{c}_{\underline{l}}(k,\beta)N^{|A_{k,\beta;\underline{l}}|}}{(2\pi i)^{2k\beta}((k\beta)!)^2}\int_{\Gamma_0}\cdots\int_{\Gamma_0}f(\underline{v};\underline{l})\hspace{-0.3cm}\sum_{\substack{t_{\sigma,\tau}=-\infty\\\text{for }1\leq\sigma<\tau\leq k}}^\infty\sum_{\substack{(x_{m,n})\\ \textbf{(}\textbf{$\star$}\textbf{)}}}\exp\left({\tfrac{1}{N}\sum x_{m,n}(v_n-v_m)^{\pm}}\right)\\
&\times\prod_{j=1}^{k-1}\delta_{N(l_j-\beta)+\sum_{\rho=j+1}^k(t_{j,\rho}+|S^-_{j,\rho}|)-\sum_{\gamma=1}^{j-1}(t_{\gamma,j}+|S^-_{\gamma,j}|)}\prod_{m=1}^{2k\beta}dv_m\nonumber.
\end{align*}
\end{lemmum}
We begin with \cref{eq:IkbetaGEO},

\begin{align*}
I_{k,\beta}(\underline{\theta})&\sim \sum_{l_1,\dots,l_{k-1}=0}^{2\beta}\frac{\tilde{c}_{\underline{l}}(k,\beta)N^{|A_{k,\beta;\underline{l}}|}}{(2\pi i)^{2k\beta}((k\beta)!)^2}\int_{\Gamma_0}\cdots\int_{\Gamma_0}f(\underline{v};\underline{l})\exp\left(iN\Big(\sum_{m=k\beta+1}^{2k\beta}\alpha_m-\beta \sum_{j=1}^k\theta_{j}\Big)\right)\\
&\quad\times\sum_{\substack{t_{\sigma,\tau}=-\infty\\\text{for }1\leq\sigma<\tau\leq k}}^\infty\exp\left(i\sum_{{1\leq}\gamma<\rho{\leq k}}(\theta_\rho-\theta_\gamma)(t_{\gamma,\rho}+|S^-_{\gamma,\rho}|)\right)\sum_{\substack{(x_{m,n})\\ \textbf{(}\textbf{$\star$}\textbf{)}}}\exp\left({\tfrac{1}{N}\sum x_{m,n}(v_n-v_m)^{\pm}}\right)\prod_{m=1}^{2k\beta}dv_m.
\end{align*}
and use the structure of the $\alpha_m$ to deduce that
\begin{equation}\label{eq:alphasum}
\exp\left(iN\Big(\sum_{m=k\beta+1}^{2k\beta}\alpha_m-\beta \sum_{j=1}^k\theta_{j}\Big)\right)=\exp\left(iN\sum_{j=1}^{k-1}(\beta-l_j)(\theta_j-\theta_k)\right).
\end{equation}

%\begin{equation}\label{eq:alphasum}
%\sum_{m=k\beta+1}^{2k\beta}\alpha_m-\beta\sum_{j=1}^k\theta_j
%&=\sum_{n=1}^{k-1}(\beta-l_n)\theta_n+(l_1+\cdots+l_{k-1}-\beta(k-1))\theta_k\\
%=\sum_{j=1}^{k-1}.
%\end{equation}

Combining \cref{eq:IkbetaGEO}, \cref{eq:momall}, \cref{eq:alphasum} and switching the order of integration we have that 
\begin{align}
\mom_N(k,\beta)&\sim\sum_{l_1,\dots,l_{k-1}=0}^{2\beta}\frac{\tilde{c}_{\underline{l}}(k,\beta)N^{|A_{k,\beta;\underline{l}}|}}{(2\pi)^k(2\pi i)^{2k\beta}((k\beta)!)^2}\int_{\Gamma_0}\cdots\int_{\Gamma_0}f(\underline{v};\underline{l})\hspace{-0.3cm}\sum_{\substack{t_{\sigma,\tau}=-\infty\\\text{for }1\leq\sigma<\tau\leq k}}^\infty\sum_{\substack{(x_{m,n})\\ \textbf{(}\textbf{$\star$}\textbf{)}}}\exp\left({\frac{1}{N}\sum x_{m,n}(v_n-v_m)^{\pm}}\right)\nonumber\\
&\times\int_0^{2\pi}\hspace{-0.3cm}\cdots\int_0^{2\pi}\exp\left({iN\sum_{j=1}^{k-1}(\beta-l_j)(\theta_j-\theta_k)}\right) \exp\left(i\sum_{\gamma<\rho}(\theta_\rho-\theta_\gamma)(t_{\gamma,\rho}+|S^-_{\gamma,\rho}|)\right) \prod_{n=1}^kd\theta_n\prod_{m=1}^{2k\beta}dv_m\label{eq:momnearlythere}.
\end{align}
By noting that $\theta_\rho-\theta_\gamma=\theta_k-\theta_\gamma-(\theta_k-\theta_\rho)$, we now see, importantly, that the $\theta$ integral will just be a function of differences $(\theta_j-\theta_k)$, $j\in\{1,\dots,k-1\}$. Focussing on the inner integral in \cref{eq:momnearlythere} we have 

%%%%%%%%%%%%%%%%%%%
%                                                     %
%                                                     %
%    Workings for the integrand       %
%                                                     %
%                                                     %
%%%%%%%%%%%%%%%%%%%
%\begin{align}
%\sum_{\sigma<\tau}(\theta_\tau-\theta_\sigma)(|S^-_{\sigma,\tau}|+t_{\sigma,\tau})&=\sum_{\sigma<\tau}(\theta_k-\theta_\sigma)(|S^-_{\sigma,\tau}|+t_{\sigma,\tau})-\sum_{\sigma<\tau}(\theta_k-\theta_\tau)(|S^-_{\sigma,\tau}|+t_{\sigma,\tau})\\
%&=\sum_{\sigma=1}^{k-1}(\theta_k-\theta_\sigma)\sum_{\tau=\sigma+1}^k(|S^-_{\sigma,\tau}|+t_{\sigma,\tau})-\sum_{\tau=2}^{k-1}(\theta_k-\theta_\tau)\sum_{\sigma=1}^{\tau-1}(|S^-_{\sigma,\tau}|+t_{\sigma,\tau})\\
%&=(\theta_k-\theta_1)\sum_{\tau=2}^k(|S^-_{\sigma,\tau}|+t_{\sigma,\tau})+\sum_{n=2}^{k-1}(\theta_k-\theta_n)\sum_{\tau=n+1}^k(|S^-_{n,\tau}|+t_{n,\tau})\\
%&\quad-\sum_{n=2}^{k-1}(\theta_k-\theta_n)\sum_{\sigma=1}^{n-1}(|S^-_{\sigma,n}|+t_{\sigma,n})\\
%&=(\theta_k-\theta_1)\sum_{\tau=2}^k(|S^-_{\sigma,\tau}|+t_{\sigma,\tau})+\sum_{n=2}^{k-1}(\theta_k-\theta_n)\left(\sum_{\tau=n+1}^k(|S^-_{n,\tau}|+t_{n,\tau})-\sum_{\sigma=1}^{n-1}(|S^-_{\sigma,n}|+t_{\sigma,n})\right)\\
%&=\sum_{n=1}^{k-1}(\theta_k-\theta_n)\left(\sum_{\tau=n+1}^k(|S^-_{n,\tau}|+t_{n,\tau})-\sum_{\sigma=1}^{n-1}(|S^-_{\sigma,n}|+t_{\sigma,n})\right).
%\end{align}
\begin{align}
\int_0^{2\pi}\cdots\int_0^{2\pi}\exp\left(i\sum_{j=1}^{k-1}(\theta_k-\theta_j)\right.&\left.\left(N(l_j-\beta)+\sum_{\rho=j+1}^k(t_{j,\rho}+|S^-_{j,\rho}|)-\sum_{\gamma=1}^{j-1}(t_{\gamma,j}+|S^-_{\gamma,j}|)\right)\right)\prod_{n=1}^kd\theta_n\nonumber\\
&\qquad\qquad=(2\pi)^k\prod_{j=1}^{k-1}\delta_{N(l_j-\beta)+\sum_{\rho=j+1}^k(t_{j,\rho}+|S^-_{j,\rho}|)-\sum_{\gamma=1}^{j-1}(t_{\gamma,j}+|S^-_{\gamma,j}|)},
\end{align}
where the $\delta$ is a Kronecker $\delta$-function.  Considering this in the context of \cref{eq:momnearlythere} we have the result.
\end{proof}

%%%%%%%%%%%%%%%%
%
%      MOM:Second lemma
%
%%%%%%%%%%%%%%%%

\begin{proof}[Proof of lemma~\ref{momlemma2}]
We restate the claim for context. 
\begin{lemmum}
\begin{align*}
&\sum_{\substack{t_{\sigma,\tau}=-\infty\\\text{for }1\leq\sigma<\tau\leq k}}^\infty\sum_{\substack{(x_{m,n})\\ \textbf{(}\textbf{$\star$}\textbf{)}}}\exp\left({\tfrac{1}{N}\sum x_{m,n}(v_n-v_m)^{\pm}}\right)\prod_{j=1}^{k-1}\delta_{N(l_j-\beta)+\sum_{\rho=j+1}^k(t_{j,\rho}+|S^-_{j,\rho}|)-\sum_{\gamma=1}^{j-1}(t_{\gamma,j}+|S^-_{\gamma,j}|)}\\%\label{eq:Nasympt1}\\
&\sim N^{|B_{k,\beta;\underline{l}}|-k+1}\kappa_k\left((k-1)\beta-\sum_{j=1}^{k-1}l_j\right)^{|B_{k,\beta;\underline{l}}|-\binom{k}{2}}\Psi_{k,\beta;\underline{l}}(((k-1)\beta-\sum_{j=1}^{k-1}l_j)\underline{v}),
\end{align*}
where $\kappa_k$ is a constant depending on $k$, \[\Psi_{k,\beta;\underline{l}}(\underline{v}) =\underset{\substack{\underline{y}=(y_{m,n})_{(m,n)\in B_{k,\beta;\underline{l}}}\\\textbf{(}\tilde{\textbf{$\ddagger$}}\textbf{)}}}{\int\cdots\int}\exp\left({\sum y_{m,n}(v_n-v_m)^{\pm}}\right)\prod dy_{m,n},\]
and $\textbf{(}\tilde{\textbf{$\ddagger$}}\textbf{)}$ denotes normalised versions of constraints \textbf{(}\textbf{$\star$}\textbf{)} { as described in the proof of lemma~\ref{lemma:geosums}}, more details given in the proof.
\end{lemmum}
Recall that the first sum in the left hand side of the statement of the lemma should be interpreted as
\[\sum_{\substack{t_{\sigma,\tau}=-\infty\\\text{for }1\leq\sigma<\tau\leq k}}^\infty=\sum_{t_{1,2}=-\infty}^\infty\sum_{t_{1,3}=-\infty}^\infty\cdots\sum_{t_{1,k}=-\infty}^\infty\sum_{t_{2,3}=-\infty}^\infty\cdots\sum_{t_{k-1,k}=-\infty}^\infty,\]
and the second sum runs over vectors $\underline{x}=(x_{m,n})_{(m,n)\in B_{k,\beta;\underline{l}}}=(x_{m,n})_{(m,n)\in \bigcup_{\sigma<\tau}S_{\sigma,\tau}}$ with integer elements subject to the following constraints 
\begin{align}
\sum_{(m,n)\in S_{1,2}}x_{m,n}&=t_{1,2}+|S_{1,2}^-|\\
\sum_{(m,n)\in S_{1,3}}x_{m,n}&=t_{1,3}+|S_{1,3}^-|\\
&\vdots\nonumber\\
\sum_{(m,n)\in S_{1,k}}x_{m,n}&=t_{1,k}+|S_{1,k}^-|\\
\sum_{(m,n)\in S_{2,3}}x_{m,n}&=t_{2,3}+|S_{2,3}^-|\\
&\vdots\nonumber\\
\sum_{(m,n)\in S_{k-1,k}}x_{m,n}&=t_{k-1,k}+|S_{k-1,k}^-|\\
H(-x_{m,n}\RE\{(v_n-v_m)^\pm\})&=1,\quad \forall (m,n)\in B_{k,\beta;\underline{l}}.
\end{align}

%We first make the change of variable $s_{\sigma,\tau}=t_{\sigma,\tau}+|S_{\sigma,\tau}^-|$ for $1\leq \sigma<\tau\leq k$ in the left hand side of the statement of the lemma, which gives
%
%\begin{align}
%&\sum_{\substack{s_{\sigma,\tau}=|S_{\sigma,\tau}^-|\\\text{for }1\leq\sigma<\tau\leq k}}^\infty\sum_{\substack{\underline{x}\\ {\textbf{(}\textbf{$\tilde{\star}$}\textbf{)}}}}\exp\left({\tfrac{1}{N}\sum x_{m,n}(v_n-v_m)^{\pm}}\right)\prod_{j=1}^{k-1}\delta_{N(l_j-\beta)+\sum_{\rho=j+1}^ks_{j,\rho}-\sum_{\gamma=1}^{j-1}s_{\gamma,j}}\label{eq:Nasympt2},
%\end{align}
%and the new constraints \textbf{(}\textbf{$\tilde{\star}$}\textbf{)} on the sum over the vector $\underline{x}$ are as follows, 
%\[\sum_{(m,n)\in S_{\sigma,\tau}}x_{m,n}=s_{\sigma,\tau},\quad\text{for }1\leq\sigma<\tau\leq k,\text{ and }{\color{red}H(-x_{m,n}\RE\{(v_n-v_m)^\pm\}=1}).\]
%
%

We now focus on the product of $\delta$-functions in the left hand side of the statement of the lemma, 

\begin{align}
&\sum_{\substack{t_{\sigma,\tau}=-\infty\\\text{for }1\leq\sigma<\tau\leq k}}^\infty\sum_{\substack{\underline{x}\\ {\textbf{(}\textbf{$\star$}\textbf{)}}}}\exp\left({\tfrac{1}{N}\sum x_{m,n}(v_n-v_m)^{\pm}}\right)\prod_{j=1}^{k-1}\delta_{N(l_j-\beta)+\sum_{\rho=j+1}^k(t_{j,\rho}+|S_{j,\rho}^-|)-\sum_{\gamma=1}^{j-1}(t_{\gamma,j}+|S_{\gamma,j}^-|)}\label{eq:Nasympt2},
\end{align}
which constrain \eqref{eq:Nasympt2} to be zero unless the following hold,

\begin{align}
\sum_{j=2}^k(t_{1,j}+|S_{1,j}^-|)&=N(\beta-l_1)\label{lconst1}\\
\sum_{j=3}^k(t_{2,j}+|S_{2,j}^-|)-(t_{1,2}+|S_{1,2}^-|)&=N(\beta-l_2)\label{lconst2}\\
&\vdots\nonumber\\
(t_{k-2,k-1}+|S_{k-2,k-1}^-|+t_{k-2,k}+|S_{k-2,k}^-|)-\sum_{j=1}^{k-3}(t_{j,k-2}+|S_{j,k-2}^-|)&=N(\beta-l_{k-2})\label{lconst3}\\
(t_{k-1,k}+|S_{k-1,k}^-|)-\sum_{j=1}^{k-2}(t_{j,k-1}+|S_{j,k-1}^-|)&=N(\beta-l_{k-1}).\label{lconst4}
\end{align}

These conditions form an underdetermined system of linear equations.  There are $\binom{k}{2}$ variables and $k-1$ equations, hence $\binom{k}{2}-(k-1)$ free parameters.  We eliminate the $k-1$
dependent variables from \eqref{eq:Nasympt2} using the linear equations \eqref{lconst1}-\eqref{lconst4}.  The following rewriting of the system picks out which we choose to discard.
\begin{align}
t_{1,k}&=N(\beta-l_1)-\sum_{j=2}^{k-1}(t_{1,j}+|S_{1,j}^-|)-|S_{1,k}^-|\label{sbound1}\\
t_{2,k}&=N(\beta-l_2)+t_{1,2}+|S_{1,2}^-|-|S_{2,k}^-|-\sum_{j=3}^{k-1}(t_{2,j}+|S_{2,j}^-|)\\
\ \ &\vdots\nonumber\\
t_{k-2,k}&=N(\beta-l_{k-2})+\sum_{j=1}^{k-3}(t_{j,k-2}+|S_{j,k-2}^-|)-(t_{k-2,k-1}+|S_{k-2,k-1}^-|+|S_{k-2,k}^-|)\\
t_{k-1,k}&=N(\beta-l_{k-1})+\sum_{j=1}^{k-2}(t_{j,k-1}+|S_{j,k-1}^-|)-|S_{k-1,k}^-|.\label{sbound2}
\end{align}

Thus, the $k-1$ outer sums over $t_{1,k},\dots,t_{k-1,k}$ in \eqref{eq:Nasympt2} collapse and we are left with

\begin{align}
&\sum_{\substack{t_{\sigma,\tau}=-\infty\\\text{for }1\leq\sigma<\tau\leq k-1}}^\infty\sum_{\substack{\underline{x}\\ \textbf{(}\textbf{$\ddagger$}\textbf{)}}}\exp\left({\tfrac{1}{N}\sum x_{m,n}(v_n-v_m)^{\pm}}\right)\label{eq:Nasympt3},
\end{align}
where $\textbf{(}\textbf{$\ddagger$}\textbf{)}$ reflects constraints on $(x_{m,n})$ as before, with \eqref{sbound1}-\eqref{sbound2} substituted in.  Since the Heaviside function $H(-x_{m,n} \RE\{(v_n-v_m)^\pm\})$ is equal to 1 for all $(m,n)\in B_{k,\beta;\underline{l}}$, this sum converges exponentially quickly.

Summing over all weights $x_{m,n}$ and using \eqref{sbound1}-\eqref{sbound2}, we have that
\begin{equation}\label{vectorconst3}
\sum_{(m,n)\in B_{k,\beta;\underline{l}}}x_{m,n}=N((k-1)\beta-\sum_{j=1}^{k-1}l_j)+\sum_{1\leq\sigma<\tau\leq k-1}(t_{\sigma,\tau}+|S_{\sigma,\tau}^-|).
\end{equation}
Clearly if any of the $t_{\sigma,\tau}$ grow faster than $N$, then there must be a least one weight $x_{m,n}^*$ having the same growth rate.  So for large $t_{\sigma,\tau}$, the summands in  \eqref{vectorconst3} will not contribute to the leading order.

In order to pull out the correct power of $N$, we employ the following general lemma about geometric sums (see Keating et al.~\cite{krrr15}, lemma 4.12).
\begin{lemma}\label{lemma:4.12}
As $K\rightarrow\infty$,
\[\sum_{\substack{k_1+\cdots+k_d=K\\k_i\geq 0}}\exp\left(\frac{1}{K}\sum_{i=1}^dk_jz_j\right)= K^{d-1}\underset{\substack{x_1+\cdots+x_d=1\\x_j\geq 0}}{\int\cdots\int}e^{\sum x_jz_j}dx_1\cdots dx_d+O(K^{d-2}).\]
\end{lemma}

From this we deduce that the leading power of $K$ in the left hand side is given by the dimension of the space described by the weights $k_j$ subject to any rules placed upon them.  Within \cref{lemma:4.12}, one can think of the weights as forming a $d$-dimensional vector where the sum of the elements must equal $K$. Thus one has $d-1$ degrees of freedom in choosing the vector elements. 

This means that the information we need to extract from the constraints given by \textbf{(}\textbf{$\ddagger$}\textbf{)} in \eqref{eq:Nasympt3} is the dimension of the space spanned by the vector $\underline{x}=(x_{m,n})_{(m,n)\in B_{k,\beta;\underline{l}}}$ subject to those restrictions. This will give us the claimed power of $N$ in the statement of lemma~\ref{momlemma2}.

Before we apply lemma~\ref{lemma:4.12}, we first incorporate $\sum_{\sigma<\tau}(t_{\sigma,\tau}+|S_{\sigma,\tau}^-|)$ into one of the weights using \eqref{vectorconst3}.  To do this, pick one weight, say $x_{1,2}$,  and shift it by $\sum_{1\leq \sigma<\tau\leq k-1}(t_{\sigma,\tau}+|S_{\sigma,\tau}^-|)$ to get an equivalent form of the inner sum in the right hand side of \eqref{eq:Nasympt3},
\begin{align}
\sum_{\substack{\underline{x}\\ \textbf{(}\textbf{$\ddagger$}\textbf{)}}}\exp&\left(\tfrac{1}{N}\sum_{\substack{(m,n)\in B_{k,\beta;\underline{l}}\\ (m,n)\neq (1,2)}}x_{m,n}(v_n-v_m)^{\pm}+\tfrac{1}{N}\big(x_{1,2}+\sum_{\sigma<\tau}(t_{\sigma,\tau}+|S_{\sigma,\tau}^-|)\big)(v_2-v_1)^{\pm}\right)\nonumber\\
&\quad=\exp\left(\tfrac{1}{N}\sum_{\sigma<\tau}(t_{\sigma,\tau}+|S_{\sigma,\tau}^-|)(v_2-v_1)^{\pm}\right)\sum_{\substack{\underline{x}\\ \textbf{(}\textbf{$\ddagger^\prime$}\textbf{)}}}\exp\left(\tfrac{1}{N}\sum_{(m,n)\in B_{k,\beta;\underline{l}}}x_{m,n}(v_n-v_m)^{\pm}\right),\label{vectorconst4}
\end{align}

The constraints \textbf{(}\textbf{$\ddagger^\prime$}\textbf{)} can be deduced from \textbf{(}\textbf{$\ddagger$}\textbf{)} by applying the described shift to the weights. In particular, this now means \cref{vectorconst3} becomes
\begin{equation}\label{sumofweights}
\sum_{(m,n)\in B_{k,\beta;\underline{l}}}x_{m,n}=N((k-1)\beta-\sum_{j=1}^{k-1}l_j).
\end{equation}

Now, we apply lemma~\ref{lemma:4.12} to the sum term on the right hand side of \eqref{vectorconst4}, with $K=N((k-1)\beta-\sum_{j=1}^{k-1}l_j)$. (The case when $(k-1)\beta=\sum_{j=1}^{k-1}l_j$ does not contribute at leading order, for reasons to be discussed at the end of the proof.) To determine the leading power of $K$, we count the amount of choice we have in choosing the weights $x_{m,n}$, for a fixed choice of $t_{\sigma,\tau}$, $1\leq \sigma<\tau\leq k-1$.  From \eqref{sbound1}-\eqref{sbound2}, we see that for each $\binom{k}{2}$ equation we lose one degree of freedom in choosing the weights, so in total, we have $|B_{k,\beta;\underline{l}}|-\binom{k}{2}$ degrees of freedom in determining $\underline{x}$.  Thus, by lemma~\ref{lemma:4.12} we have 
\begin{align}
\sum_{\substack{\underline{x}\\ \textbf{(}\textbf{$\ddagger^\prime$}\textbf{)}}}&\exp\Big(\tfrac{1}{N}\sum_{(m,n)\in B_{k,\beta;\underline{l}}}x_{m,n}(v_n-v_m)^{\pm}\Big)\nonumber\\
&=\Big(N((k-1)\beta-\sum_{j=1}^{k-1}l_j)\Big)^{|B_{k,\beta;\underline{l}}|-\binom{k}{2}}\Psi_{k,\beta;\underline{l}}(((k-1)\beta-\sum_{j=1}^{k-1}l_j)\underline{v})+O\big(N^{|B_{k,\beta;\underline{l}}|-\binom{k}{2}-1}\big)\label{eq:Nasympt7},
\end{align}
where \[\Psi_{k,\beta;\underline{l}}(\underline{v}) =\underset{\substack{\underline{y}=(y_{m,n})_{(m,n)\in B_{k,\beta;\underline{l}}}\\\textbf{(}\tilde{\textbf{$\ddagger$}}\textbf{)}}}{\int\cdots\int}\exp\left({\sum y_{m,n}(v_n-v_m)^{\pm}}\right)\prod dy_{m,n},\]
and $\textbf{(}\tilde{\textbf{$\ddagger$}}\textbf{)}$ denotes the normalised version of the constraints $\textbf{(}{\textbf{$\ddagger^\prime$}}\textbf{)}$ (since we need only consider $t_{\sigma,\tau}$ growing at most like $N$ for each $1\leq\sigma<\tau\leq k-1$, asymptotically the constraints $\textbf{(}\tilde{\textbf{$\ddagger$}}\textbf{)}$ will be $O(1)$, and in particular will not depend on $t_{\sigma,\tau}$).  

Then, combining \cref{eq:Nasympt7} with \cref{vectorconst4} and \cref{eq:Nasympt3}, we have
\begin{align}
&\sum_{\substack{t_{\sigma,\tau}=-\infty\\\text{for }1\leq\sigma<\tau\leq k-1}}^\infty\sum_{\substack{\underline{x}\\ \textbf{(}\textbf{$\ddagger$}\textbf{)}}}\exp\left({\tfrac{1}{N}\sum x_{m,n}(v_n-v_m)^{\pm}}\right)\nonumber\\
%&\sum_{\substack{|S_{\sigma,\tau}^-|\leq s_{\sigma,\tau}\leq S\\\text{for }1\leq\sigma<\tau\leq k-1\\\textbf{(}\textbf{${\dagger}$}\textbf{)}}}^{\text{ORDER N}}\sum_{\substack{\underline{x}\\ \textbf{(}\textbf{$\ddagger$}\textbf{)}}}\exp\left({\tfrac{1}{N}\sum x_{m,n}(v_n-v_m)^{\pm}}\right)\nonumber\\
&\sim \left(N((k-1)\beta-\sum_{j=1}^{k-1}l_j)\right)^{|B_{k,\beta;\underline{l}}|-\binom{k}{2}}\nonumber\\
&\qquad\times  \sum_{\substack{t_{\sigma,\tau}=-\infty\\\text{for }1\leq\sigma<\tau\leq k-1}}^\infty \Psi_{k,\beta;\underline{l}}(((k-1)\beta-\sum_{j=1}^{k-1}l_j)\underline{v}) \prod_{1\leq\sigma<\tau\leq k-1}\exp\left(\tfrac{1}{N}(t_{\sigma,\tau}+|S_{\sigma,\tau}^-|)(v_2-v_1)^{\pm}\right)\\
%\sum_{\substack{s_{\sigma,\tau}=|S_{\sigma,\tau}^-|\\\text{for }1\leq\sigma<\tau\leq k-1\\\textbf{(}\textbf{${\dagger}$}\textbf{)}}}^{\text{ORDER N}}\exp\left(-\tfrac{1}{N}\sum_{\sigma<\tau}s_{\sigma,\tau}(v_2-v_1)^{\pm}\right)\\
%&\qquad\sim \left(N((k-1)\beta-\sum_{j=1}^{k-1}l_j)\right)^{|B_{k,\beta;\underline{l}}|-\binom{k}{2}}\Psi_{k,\beta;\underline{l}}(((k-1)\beta-\sum_{j=1}^{k-1}l_j)\underline{v})\\
%&\qquad\qquad\qquad\times \prod_{1\leq\sigma<\tau\leq k-1}\sum_{\substack{s_{\sigma,\tau}=|S_{\sigma,\tau}^-|\\\textbf{(}\textbf{${\dagger}$}\textbf{)}}}^{\text{ORDER N}}\exp\left(-\tfrac{1}{N}s_{\sigma,\tau}(v_2-v_1)^{\pm}\right).
&\sim \kappa_k N^{\binom{k}{2}-(k-1)}\left(N((k-1)\beta-\sum_{j=1}^{k-1}l_j)\right)^{|B_{k,\beta;\underline{l}}|-\binom{k}{2}}\Psi_{k,\beta;\underline{l}}(((k-1)\beta-\sum_{j=1}^{k-1}l_j)\underline{v}),
\end{align}

for some constant $\kappa_k$ depending on $k$. Note that the case where $\sum_{j=1}^{k-1}l_j=(k-1)\beta$ falls in to the subleading order terms.

Thus at leading order,

\begin{align}
&\sum_{\substack{t_{\sigma,\tau}=-\infty\\\text{for }1\leq\sigma<\tau\leq k}}^\infty\sum_{\substack{(x_{m,n})\\ \textbf{(}\textbf{$\star$}\textbf{)}}}\exp\left({\tfrac{1}{N}\sum x_{m,n}(v_n-v_m)^{\pm}}\right)\prod_{j=1}^{k-1}\delta_{N(l_j-\beta)+\sum_{\rho=j+1}^k(t_{j,\rho}+|S^-_{j,\rho}|)-\sum_{\gamma=1}^{j-1}(t_{\gamma,j}+|S^-_{\gamma,j}|)}\nonumber\\
&\sim N^{|B_{k,\beta;\underline{l}}|-k+1}\kappa_k \left((k-1)\beta-\sum_{j=1}^{k-1}l_j\right)^{|B_{k,\beta;\underline{l}}|-\binom{k}{2}}\Psi_{k,\beta;\underline{l}}(((k-1)\beta-\sum_{j=1}^{k-1}l_j)\underline{v}).\label{finalasympt}
\end{align}

\end{proof}

%%%%%%%%%%%%%%%%
%
%      MOM:Third lemma
%
%%%%%%%%%%%%%%%%

\begin{proof}[Proof of lemma~\ref{momlemma3}]

We restate the lemma.

\begin{lemmum}
\[\mom_N(k,\beta)\sim \gamma_{k,\beta}N^{k^2\beta^2-k+1}\]
where 
\begin{equation*}\label{eq:gamma}
\gamma_{k,\beta}= \sum_{l_1,\dots,l_{k-1}=0}^{2\beta}{c}_{k,\beta;\underline{l}}((k-1)\beta-\sum_{j=1}^{k-1}l_j)^{|B_{k,\beta;\underline{l}}|-\binom{k}{2}}P_{k,\beta}(l_1,\dots,l_{k-1}),
\end{equation*}
${c}_{k,\beta;\underline{l}}$ is a constant depending on $k,\beta,l_1,\dots,l_{k-1}$ (more details given below), and
\begin{align*}
P_{k,\beta}(l_1,\dots,l_{k-1})&=\frac{(-1)^{g(k,\beta;\underline{l})}}{(2\pi i)^{2k\beta}((k\beta)!)^2}\\
&\times\int_{\Gamma_0}\cdots\int_{\Gamma_0}\frac{e^{-\sum_{m=k\beta+1}^{2k\beta}v_m}\prod_{\substack{m<n\\\alpha_m=\alpha_n}}\left({v_n-v_m}\right)^2}{\prod_{\substack{m\leq k\beta<n\\\alpha_m=\alpha_n}}\left({v_n-v_m}\right)\prod_{m=1}^{2k\beta}v_m^{2\beta}}\Psi_{k,\beta;\underline{l}}(((k-1)\beta-\sum_{j=1}^{k-1}l_j)\underline{v})\prod_{m=1}^{2k\beta}dv_m,\nonumber
\end{align*}
with $\Psi_{k,\beta;\underline{l}}(\underline{v})$ as defined in \cref{momlemma2}, and $g(k,\beta;\underline{l})$ given by \eqref{overallphase}.
\end{lemmum}

Recall the definition of the sets $A_{k,\beta;\underline{l}}$ and $B_{k,\beta;\underline{l}}$,
\begin{align}
A_{k,\beta;\underline{l}}&\coloneqq\{(m,n):1\leq m\leq k\beta<n\leq 2k\beta, \alpha_m=\alpha_n\}\\
B_{k,\beta;\underline{l}}&\coloneqq\{(m,n):1\leq m\leq k\beta<n\leq 2k\beta, \alpha_m\neq\alpha_n\},
\end{align}
so $|A_{k,\beta;\underline{l}}|+|B_{k,\beta;\underline{l}}|=k^2\beta^2$.  Using this fact with \cref{momlemma1} and \cref{momlemma2} we have  

\begin{align}
\mom_N(k,\beta)&\sim N^{k^2\beta^2-k+1}\sum_{l_1,\dots,l_{k-1}=0}^{2\beta}\frac{(-1)^{g(k,\beta;\underline{l})}{c}_{k,\beta;\underline{l}}((k-1)\beta-\sum_{j=1}^{k-1}l_j)^{{|B_{k,\beta;\underline{l}}|-\binom{k}{2}}}}{(2\pi i)^{2k\beta}((k\beta)!)^2}\nonumber\\
&\quad\times\int_{\Gamma_0}\cdots\int_{\Gamma_0}f(\underline{v};\underline{l})\Psi_{k,\beta;\underline{l}}\Big(\Big((k-1)\beta-\sum_{j=1}^{k-1}l_j\Big)\underline{v}\Big)\prod_{m=1}^{2k\beta}dv_m,
\end{align}
where ${c}_{k,\beta;\underline{l}}$ is a constant encompassing the two constants given in \eqref{prodofbinom} and the statement of \cref{momlemma2},

Lemma~\ref{momlemma3} then follows recalling the definition of $f(\underline{v};\underline{l})$, 
\[f(\underline{v};\underline{l})=\frac{e^{-\sum_{m=k\beta+1}^{2k\beta}v_m}\prod_{\substack{m<n\\\alpha_m=\alpha_n}}\left({v_n-v_m}\right)^2}{\prod_{\substack{m\leq k\beta<n\\\alpha_m=\alpha_n}}\left({v_n-v_m}\right)\prod_{m=1}^{2k\beta}v_m^{2\beta}},\]
and by setting $\gamma_{k,\beta}$ and $P_{k,\beta}(l_1,\dots,l_{k-1})$ as claimed. 

\end{proof}

%%%%%%%%%%%%%%%%
%
%      MOM:Fourth lemma
%
%%%%%%%%%%%%%%%%

\begin{proof}[Proof of lemma~\ref{momlemma4}]
Finally, recall the statement of lemma~\ref{momlemma4}.

\begin{lemmum}
For $k,\beta\in\mathbb{N}$, $\gamma_{k,\beta}\neq0$ where $\gamma_{k,\beta}$ is as defined in lemma~\ref{momlemma3}. 
\end{lemmum}

Thus, we have to show that 
\begin{equation}\label{eq:gamma1}
 \sum_{l_1,\dots,l_{k-1}=0}^{2\beta}{c}_{k,\beta;\underline{l}}((k-1)\beta-\sum_{j=1}^{k-1}l_j)^{|B_{k,\beta;\underline{l}}|-\binom{k}{2}}P_{k,\beta}(l_1,\dots,l_{k-1})\neq 0,
\end{equation}
for ${c}_{k,\beta;\underline{l}}$ some constant depending on $k,\beta, \underline{l}$ and
\begin{align}
P_{k,\beta}(l_1,\dots,l_{k-1})&=\frac{(-1)^{g(k,\beta;\underline{l})}}{(2\pi i)^{2k\beta}((k\beta)!)^2}\nonumber\\
&\times\int_{\Gamma_0}\cdots\int_{\Gamma_0}\frac{e^{-\sum_{m=k\beta+1}^{2k\beta}v_m}\prod_{\substack{m<n\\\alpha_m=\alpha_n}}\left({v_n-v_m}\right)^2}{\prod_{\substack{m\leq k\beta<n\\\alpha_m=\alpha_n}}\left({v_n-v_m}\right)\prod_{m=1}^{2k\beta}v_m^{2\beta}}\Psi_{k,\beta;\underline{l}}(((k-1)\beta-\sum_{j=1}^{k-1}l_j)\underline{v})\prod_{m=1}^{2k\beta}dv_m,\label{vandermondeintegrand}
\end{align}
Since $c_{k,\beta;\underline{l}}$ is a constant encompassing both \eqref{prodofbinom} and the constant appearing in the statement of  \cref{momlemma2}, it is clearly non-zero and its sign is independent of the sum parameters $l_1,\dots,l_{k-1}$. Further, at leading order we need only consider parameters $l_j$ such that $l_1+\cdots+l_{k-1}\neq (k-1)\beta$.  To prove the required result, we show that  $((k-1)\beta-\sum_{j=1}^{k-1}l_j)^{|B_{k,\beta;\underline{l}}|-\binom{k}{2}}P_{k,\beta}(\underline{l})\neq 0$ (in fact, it is strictly positive) for such $l_1,\dots,l_{k-1}$. To do this, we will appeal to the residue theorem. 

Fix a choice of $l_1,\dots,l_{k-1}$ in agreement with the various constraints.  To show that $P_{k,\beta}(\underline{l})$ is non-zero, { firstly denote the integrand in (\ref{vandermondeintegrand}) by $q_{k,\beta}(l_1,\dots,l_{k-1})$.  Then by the residue theorem} we have to show that there is a term of the form $(v_1\cdots v_{2k\beta})^{2\beta-1}$ with non-zero coefficient in the expansion of 
\begin{equation}\label{eq:residuetheorem1}
{q_{k,\beta}(l_1,\dots,l_{k-1})\prod_{m=1}^{2k\beta}v_m^{2\beta}=\Psi_{k,\beta;\underline{l}}\big(\big((k-1)\beta-\sum_{j=1}^{k-1}l_j\big)\underline{v}\big)e^{-\sum_{m=k\beta+1}^{2k\beta}v_m}\frac{\prod_{\substack{m<n\\\alpha_m=\alpha_n}}(v_n-v_m)^2}{\prod_{\substack{m\leq k\beta<n\\\alpha_m=\alpha_n}}(v_n-v_m)}}.
\end{equation}

{Now, simplifying the product terms of the right hand side of $(\ref{eq:residuetheorem1})$, }
\begin{align}
\frac{{ \prod_{\substack{m<n\\\alpha_m=\alpha_n}}(v_n-v_m)^2}}{{ \prod_{\substack{m\leq k\beta<n\\\alpha_m=\alpha_n}}(v_n-v_m)}}&{=\prod_{\substack{m\leq k\beta<n\\\alpha_m=\alpha_n}}\left({v_n-v_m}\right)\prod_{n=1}^{k}\Delta(v_{\sum_{j=1}^{n-1}l_j+1},\dots,v_{\sum_{j=1}^{n}l_j})^2}\nonumber\\
&{\times \prod_{n=1}^{k}\Delta(v_{\sum_{j=1}^{n}l_j+2(k-n)\beta+1},\dots,v_{\sum_{j=1}^{n-1}l_j+2(k-(n-1))\beta})^2,}\label{expansion}
\end{align}
where $l_k=k\beta-(l_1+\cdots+l_{k-1})$. We use the following expansion of the Vandermonde determinant,
\[\Delta(x_1,\dots,x_n)^2=\sum_{\sigma,\tau\in S_n}\Sgn(\sigma)\Sgn(\tau)\prod_{i=1}^nx_i^{\sigma(i)+\tau(i)-2}.\]

From any of the terms appearing in the first product of Vandermonde determinants in the right hand side of \cref{expansion}, we find a term of the form \[l_n!\prod_{i=\sum_{j=1}^{n-1}l_j+1}^{\sum_{j=1}^{n}l_j}v_i^{l_n-1},\quad n\in\{1,\dots,k\},\] and similarly for any of the terms in the second product.  Thus,  the Vandermonde determinants collectively contribute a term of the form 
\begin{equation}\label{eq:vdmterm}
\prod_{i=1}^{l_1}v_i^{l_1-1}\prod_{i=l_1+1}^{l_1+l_2}v_i^{l_2-1}\cdots \prod_{i=\sum_{j=1}^{k-1}l_j+1}^{k\beta}v_i^{l_k-1}\prod_{i=k\beta+1}^{\sum_{j=1}^{k-1}l_j+2\beta}v_i^{2\beta-l_k-1}\cdots\prod_{i=2(k-1)\beta+1+l_1}^{2k\beta}v_i^{2\beta-l_1-1},
\end{equation}
and this term has a strictly positive coefficient (a detailed explanation can be found in the Appendix, see~\cref{app:scott}) given by

\[\prod_{j=1}^kl_j!(2\beta-l_j)!.\] 

We expand the remaining product as 
\begin{equation}\label{productexpansion}
\prod_{\substack{m\leq k\beta<n\\\alpha_m=\alpha_n}}\left({v_n-v_m}\right)=\prod_{\substack{m\in\{1,\dots,l_1\}\\n\in\{l_1+2(k-1)\beta+1,\dots,2k\beta\}}}(v_n-v_m)\cdots\prod_{\substack{m\in\{\sum_{j=1}^{k-1}l_j+1,\dots,k\beta\}\\n\in\{k\beta+1,\dots,\sum_{j=1}^{k-1}l_j+2\beta\}}}(v_n-v_m).
\end{equation}
From the first product in the right hand side of \cref{productexpansion} we take the term $\prod_{i=1}^{l_1}(-v_i)^{2\beta-l_1}$.  The second gives $\prod_{i=l_1+1}^{l_1+l_2}(-v_i)^{2\beta-l_2}$, and so on.  Hence, in total from \cref{productexpansion} we have a term of the form 
\begin{equation}\label{middlecoeff}
(-1)^{k\beta}\prod_{j=1}^{k}\left(\prod_{i=\sum_{n=1}^{j-1}l_n+1}^{\sum_{n=1}^{j}l_n}v_i^{2\beta-1}\right)\prod_{j=1}^{k}\left(\prod_{i=\sum_{n=1}^{j}l_n+2(k-j)\beta+1}^{\sum_{n=1}^{j-1}l_n+2(k-(j-1))\beta}v_i^{2\beta-l_j-1}\right).
\end{equation}
We now use the exponential function in \cref{eq:residuetheorem1} to give us the remaining contribution,
\begin{align}
e^{-\sum_{m=k\beta+1}^{2k\beta}v_m}&=\sum_{t=0}^\infty \frac{\left(-\sum_{m=k\beta+1}^{2k\beta}v_m\right)^t}{t!}\\
&=\sum_{t=0}^\infty \frac{(-1)^t}{t!}\sum_{a_{k\beta+1}+\dots+a_{2k\beta}=t}\binom{t}{a_{k\beta+1},\dots,a_{2k\beta}}\prod_{i=k\beta+1}^{2k\beta}v_i^{a_i},
\end{align}
where the multinomial coefficient is \[\binom{n}{c_1,\dots,c_m}=\frac{n!}{c_1!\cdots c_m!}.\]
To complete the construction of the term of the form $(v_1\cdots v_{2k\beta})^{2\beta-1}$, we need 
\begin{equation}
a_{i}=\begin{cases}
l_k&\text{for }i\in\{k\beta+1,\dots,\sum_{j=1}^{k-1}l_j+2\beta\}\\
l_{k-1}&\text{for }i\in\{\sum_{j=1}^{k-1}l_j+2\beta+1,\dots,\sum_{j=1}^{k-2}l_j+4\beta\}\\
\vdots&\quad\vdots\\
l_1&\text{for }i\in\{l_1+2(k-1)\beta+1,\dots,2k\beta\}.
\end{cases}
\end{equation}

Hence the required coefficient comes from looking at the term for which $t=\sum_i a_i=\sum_i l_i(2\beta-l_i)$, which has coefficient 
\begin{equation}\label{finalcoeff}
\frac{(-1)^{\sum_{i=1}^k l_i(2\beta-l_i)}\binom{\sum l_i(2\beta-l_i)}{l_k,\dots,l_k,\dots,l_1,\dots,l_1}}{\left(\sum_{i=1}^k l_i(2\beta-l_i)\right)!}=\frac{(-1)^{k\beta}}{(l_1!)^{2\beta-l_1}\cdots (l_k!)^{2\beta-l_k}}.
\end{equation}

Thus, we have constructed a term of the form $(v_1\cdots v_{2k\beta})^{2\beta-1}$ which has strictly positive coefficient (the prefactors of $(-1)^{k\beta}$ in \cref{middlecoeff} and \cref{finalcoeff} cancel each other) given by
\begin{equation}
\frac{\prod_{j=1}^kl_j!(2\beta-l_j!)}{(l_1!)^{2\beta-l_1}\cdots (l_k!)^{2\beta-l_k}}.
\end{equation}

In fact this is the only way to construct a term of this form from the integrand (more details are given in the Appendix, see \cref{sec:uniqueconstruction}). 

All that is left to prove is that the term 
\[(-1)^{g(k,\beta;\underline{l})}((k-1)\beta-\sum_{j=1}^{k-1}l_j)^{|B_{k,\beta;\underline{l}}|-\binom{k}{2}}\Psi_{k,\beta;\underline{l}}(((k-1)\beta-\sum_{j=1}^{k-1} l_j)\underline{v})\]
only contributes a positive coefficient, where recall

\begin{align}
&\Psi_{k,\beta;\underline{l}}(((k-1)\beta-\sum_{j=1}^{k-1} l_j)\underline{v})=\nonumber\\
&\qquad\underset{\substack{\underline{y}=(y_{m,n})_{(m,n)\in B_{k,\beta;\underline{l}}}\\\textbf{(}\tilde{\textbf{$\ddagger$}}\textbf{)}}}{\int\cdots\int}\hspace{-.3cm}\exp\left(((k-1)\beta-\sum_{j=1}^{k-1} l_j)\sum y_{m,n}(v_n-v_m)^\pm\right)\prod dy_{m,n}.\label{psifunction}
\end{align}

Calculating $\gamma_{k,\beta}$ involves computing derivatives of \cref{psifunction} and evaluating it at $\underline{v}=0$ by the residue theorem.  
%\begin{align}
%\prod_{(m,n)\in B_{k,\beta;\underline{l}}}\exp\left(y_{m,n}(v_n-v_m)^\pm\right)&=\prod_{\sigma<\tau}\prod_{(m,n)\in S_{\sigma,\tau}}\exp\left(y_{m,n}(v_n-v_m)^\pm\right)\\
%&=\prod_{\sigma<\tau}\prod_{(m,n)\in S_{\sigma,\tau}^-}\exp\left(-y_{m,n}(v_n-v_m)\right)\prod_{(m,n)\in S_{\sigma,\tau}^+}\exp\left(y_{m,n}(v_n-v_m)\right)
%\end{align}
We consider the case where $(k-1)\beta>l_1+\cdots+l_{k-1}$, the other case follows similarly. Incorporating in the sign $(-1)^{g(k,\beta;\underline{l})}$ in to the integrand of the right hand side of \eqref{psifunction} (where for simplicity we ignore the positive prefactor in the exponent since it doesn't contribute to the overall sign)  we have
\begin{align}
(-1)^{g(k,\beta;\underline{l})}\prod_{(m,n)\in B_{k,\beta;\underline{l}}}\exp\left(y_{m,n}(v_n-v_m)^\pm\right)&=\prod_{\sigma<\tau}\prod_{(m,n)\in S_{\sigma,\tau}^-}(\operatorname{sign}(\RE\{(v_m-v_n)\})\exp\left(-y_{m,n}(v_n-v_m)\right))\nonumber\\
&\times\prod_{(m,n)\in S_{\sigma,\tau}^+}(-\operatorname{sign}(\RE\{(v_n-v_m)\})\exp\left(y_{m,n}(v_n-v_m)\right)).\label{overallsgn}
\end{align}

Thus, in order to show that $\gamma_{k,\beta}$ is strictly positive we need to establish that differentiating the right hand side of \cref{overallsgn} contributes an overall positive sign.  To see that this is true, first note that since each of the $v_m$, for $m\in\{1,\dots,2k\beta\}$, in $P_{k,\beta}(\underline{l})$ has a pole of even order at $0$, and by the residue theorem we are required to differentiate the exponential term in \cref{psifunction} an odd number of times.  Then, by the requirements of the conditions on the Riemann integral in the right hand side of \cref{psifunction}, for each $(m,n)$, one has that the Heaviside function ensures that the product $y_{m,n}\RE\{(v_n-v_m)^\pm\}$ is negative.  It is easy to check that in each case, after differentiating an odd number of times, that the term on the right hand of \cref{overallsgn} is positive. This concludes the proof of lemma~\ref{momlemma4}.
\end{proof}

%%%%%%%%%%%%%%%%%%%                                                     
%                                                     %
%                                                     %
%           Polynomiality                     %
%                                                     %
%                                                     %
%%%%%%%%%%%%%%%%%%%

\section{Polynomial Structure}
\label{sec:polystruct}

In this section we prove theorem~\ref{thm:polynomial}. The technique we use relies on a formula for $I_{k,\beta}(\theta_1,\dots,\theta_k)$ (c.f.~(\ref{eq:pureca}) and (\ref{eq:momall})) that follows from an expression obtained by Conrey et al.~\cite{cfkrs1}.  This takes the form of a combinatorial sum and is a special case of a more general expression that we state in the Appendix.

\begin{theorem}\label{ratiothm}
Let $\Xi_{k\beta}$ be the set of $\binom{2k\beta}{k\beta}$ permutations $\sigma\in S_{2k\beta}$ such that $\sigma(1)<\sigma(2)<\cdots<\sigma(k\beta)$ and $\sigma(k\beta+1)<\cdots<\sigma(2k\beta)$, and 
\[\underline{\omega}=(\underbrace{e^{i\theta_1},\dots,e^{i\theta_1}}_{\beta},\dots,\underbrace{e^{i\theta_k},\dots,e^{i\theta_k}}_{\beta},\underbrace{e^{i\theta_1},\dots,e^{i\theta_1}}_{\beta},\dots,\underbrace{e^{i\theta_k},\dots,e^{i\theta_k}}_{\beta}).\]

Then, 
\[\left(\prod_{j=k\beta+1}^{2k\beta}\omega_j^{N}\right) I_{k,\beta}(\theta_1,\dots,\theta_k)=\sum_{\sigma\in \Xi_{k\beta}}\frac{(\omega_{\sigma(k\beta+1)}\omega_{\sigma(k\beta+2)}\cdots\omega_{\sigma(2k\beta)})^N}{\prod_{l\leq k\beta<q}(1-\omega_{\sigma(l)}\omega^{-1}_{\sigma(q)})}.\]
\end{theorem}

Therefore,
\begin{equation}
\mom_N(k,\beta)=\frac{1}{(2\pi)^k}\int_0^{2\pi}\cdots\int_0^{2\pi}\prod_{j=k\beta+1}^{2k\beta}\omega_j^{-N}\sum_{\sigma\in \Xi_{k\beta}}\frac{(\omega_{\sigma(k\beta+1)}\omega_{\sigma(k\beta+2)}\cdots\omega_{\sigma(2k\beta)})^N}{\prod_{l\leq k\beta<q}(1-\omega_{\sigma(l)}\omega^{-1}_{\sigma(q)})}d\theta_1\cdots d\theta_k.
\end{equation}

The individual summands in the integrand in this expression have poles of finite order (when $\omega_{\sigma(q)}=\omega_{\sigma(l)}$).  These cancel with zeros in the numerator in the complete sum, as they must because $I_{k,\beta}(\theta_1,\dots,\theta_k)$ is bounded, being the average of a product of polynomials~\cite{cfkrs1}.   The function remaining after this cancellation may be computed by applying l'H\^opital's rule a finite number of times.  This function is therefore a polynomial in the variables $e^{i\theta_1},\dots,e^{i\theta_k}$ with coefficients that are each polynomial functions of $N$ (coming from the derivatives associated with applying l'H\^opital's rule).  Upon integrating only the coefficient of the constant term remains, which is polynomial in $N$.    This concludes the proof of $\ref{thm:polynomial}$.  In principle one could compute the order of the polynomial this way, but in general we found the approach based on the asymptotic evaluation of the integral representation, set out in the previous section, to be more straightforward.  In specific cases the calculation is feasible, as demonstrated in the Appendix.   

%%%%%%%%%%%%%%%%%%%                                                     
%                                                     %
%                                                     %
%                 Outlook                        %
%                                                     %
%                                                     %
%%%%%%%%%%%%%%%%%%%
\section{Summary and Outlook}
\label{Summary and Outlook}

Our main result is a proof that the moments of the moments of the characteristic polynomials of random unitary matrices, 
$\mom_N(k,\beta)$, are polynomial functions of $N$, of order $k^2\beta^2-k+1$,  when $k$ and $\beta$ both take values in $\mathbb{N}$.  This proves the conjecture for the leading order asymptotics made in \cite{fyodorov12, fyodorov14} when $k$ and $\beta$ both take values in $\mathbb{N}$.  Moreover, it goes further in establishing that an exact formula exists when $k$ and $\beta$ both take values in $\mathbb{N}$, and, in passing, establishes the general structure of the (finite) asymptotic expansion for $\mom_N(k,\beta)$ in this case.

It is clear from the calculation set out in Section 2 that we have an exact formula when $k$ and $\beta$ both take values in $\mathbb{N}$ because of an underlying integrable structure: the approach based on symmetric function theory, and hence on representation theory, yields an exact formula in terms of a count of certain restricted semistandard Young tableaux.  The symmetric functions used in Section 2 may be related to certain generalized hypergeometric functions (c.f.~\cite{macdonald98} and, for example~\cite{FK04}), and it would be interesting to explore this calculation in that context, especially if doing so extends the results to non-integer values of $k$ and $\beta$.  We see our calculation as a first step in that direction and anticipate pursuing this further.  We note in passing that the formula we establish using the multiple integral approach provides as a byproduct an asymptotic expression for the count of semistandard Young tableaux that arises in the calculation.  

The moments of the moments we study here play a central role in the heuristic analysis in \cite{fyodorov12, fyodorov14, Keating_lec} of the value distribution of $\log P_{\rm max}(A)$, leading to the conjecture that as $N\rightarrow\infty$
\begin{equation}
\label{FKconjecture}
 \log P_{\rm max}(A)=\log N -\frac{3}{4}\log\log N +x_N(A),
 \end{equation}
where $x_N(A)$ is a random variable that is $O_\mathbb{P}(1)$ and which has a limiting value distribution that is a sum of two Gumbel distributions.
Several components of these conjectures have recently been proved: the first term on the righthand side of (\ref{FKconjecture}) was established in \cite{ABB}, the second term in \cite{PZ2017}, and the tightness of $x_N(A)$ in \cite{CMN}.  All of these calculations have utilised a hierarchical branching structure in the Fourier expansion of $\log |P_N(A,\theta)|$,
\begin{equation}
\label{Fourier}
 \log |P_N(A,\theta)|=-{\rm Re}\sum_{k=1}^\infty \frac{{\rm Tr}A^k}{k}\exp (ik\theta),
 \end{equation}
similar to that found in other log-correlated Gaussian fields such as the branching random walk and the two-dimensional Gaussian Free Field; that is, they have utilised general probabilistic methods.  When $\log |P_N(A,\theta)|$ (c.f.~(\ref{Fourier})) is replaced by a random Fourier series with the same correlation structure -- such series can be considered as one-dimensional models of the two-dimensional Gaussian Free Field -- the analogue of conjecture (\ref{conj}), due to Fyodorov and Bouchaud \cite{FB2008}, has recently been proved for all $k$ and $\beta$ in the regime $k<1/\beta^2$ by Remy \cite{Remy1} using ideas from conformal field theory \cite{KRV}.

Formally, the $\beta\rightarrow\infty$ asymptotics of $Z_N(A, \beta)$ determines $P_{\rm max}(A)$, and so it is natural to seek to understand the value distribution of $P_{\rm max}(A)$ by calculating the moments of $Z_N(A, \beta)$ and then taking the large-$\beta$ limit.  However, doing this requires the moments for all $k$ and $\beta$, not just the integer moments.  Moreover, the controlling range is when freezing dominates and $k\beta^2$ is close to 1.  Our results therefore cannot be applied as they stand.  This is one reason why the possibility of using the integrable structure to extend them to non-integer values of $k$ and $\beta$ is attractive.  When $k=1$ the Selberg integral makes this possible.  (And in the somewhat similar problem of the joint moments of the characteristic polynomial and its derivative, Painlev\'e theory provides a route (c.f.~\cite{BBBGIIK, BBBCPRS}).) 

The association between characteristic polynomials of random matrices and the theory of the Riemann zeta-function motivates analogous conjectures to those just described for
 \begin{equation}
 \zeta_{\rm max}(T)=\max_{0\le x <2\pi}|\zeta(1/2+iT+ix)|,
 \end{equation}
where $T$ is random \cite{fyodorov12, fyodorov14, Keating_lec}.  These correspond to replacing $N$ in (\ref{FKconjecture}) by $\log T$ (c.f~\cite{keasna00a}).  In this case too there has been recent progress in proving the leading order term in the resulting formula when $T\rightarrow\infty$ \cite{Najnudel, ABBRS}, based on calculations that mirror those for the extremes of characteristic polynomials.  

The multiple-integral approach we have developed here also applies to the zeta-function, using the representation established in \cite{cfkrs2}, giving explicit (conjectural) formulae for the integer moments of the integer moments over short intervals of the critical line in that case too.  These take the form of polynomials in $\log T$ up to an error that is a power of $T$ smaller.  This is important because in numerical computations of the moments one is necessarily restricted to finite intervals, and it is a key question how moments computed in different intervals fluctuate.  Our formula for the moments gives an answer to this question.  We intend to discuss this in more detail in a forthcoming paper.  

The methods of calculation we have developed here for moments defined with respect to averages over the unitary group extend to the other classical compact groups: both the representations in terms of symmetric functions and multiple integrals have been developed \cite{cfkrs1, bumgam06}.  The extension of our results to the associated random matrix ensembles can therefore be worked out in an analogous way to that described in this paper.  This would then have applications to the other symmetry classes of $L$-functions, using \cite{cfkrs2}, in a similar way as for the Riemann zeta-function.  Our results apply immediately (and unconditionally) as well to the moments of the moments of function field $L$-functions defined over $\mathbb{F}_q$ in the limit $q\to\infty$.  This follows from equidistribution results in that case (c.f.~\cite{Keating_lec}).

Finally, our formulae have already been applied to analysing the results of numerical computations using randomly generated unitary matrices, where they explain the fluctuations in the moments of the characteristic polynomials evaluated by averaging over the unit circle \cite{FGK}.  We anticipate further similar applications and extensions to other numerical computations of the moments of spectral determinants.

%%%%%%%%%%%%%%%%%%%
%                                                     %
%                                                     %
%           Appendix                           %
%                                                     %
%                                                     %
%%%%%%%%%%%%%%%%%%%
\newpage
\section{Appendix}

\subsection{Examples}\label{app:euan}

Here we give explicit examples of the polynomials $\mom_N(k,\beta)$ for small values of $k,\beta$.  The formulae we record extend the results of preliminary calculations due to Keating and Scott \cite{KeatingScott} (c.f.~\cite{Keating_lec}), which formed the basis for some of the numerical computations in \cite{FGK}.

First, the general technique is described and then explicit forms of the polynomials are given in the cases of $\beta=1$, $k\in\{1,2,3,4\}$ and $\beta=2$, $k\in\{1,2\}$. We should remark that the moment formula of Keating and Snaith~\cite{keasna00a} gives the full polynomials for the case $k=1$, $\beta\in \mathbb{N}$; see (\ref{k=1}).

The technique we use is in a slightly more general form than is needed here, because we see it as having other potential applications; we then specialise back to the actual formula required for our calculations.  The more general form we start with was first derived by Conrey, Farmer and Zirnbauer \cite{CFZ}, and later by Bump and Gamburd using symmetric function theory \cite{bumgam06}. Note that we used a special case of this result to prove theorem~\ref{thm:polynomial} in Section~\ref{sec:polystruct}. First, define for finite sets $A,B,C,D$,
\[R(A,B;C,D)\coloneqq\int_{U(N)}\frac{\prod_{\alpha\in A}\det(I-X^*e^{-\alpha})\prod_{\beta\in B}\det(I-Xe^{-\beta})}{\prod_{\gamma\in C}\det(I-X^*e^{-\gamma})\prod_{\delta\in D}\det(I-Xe^{-\delta})}dX.\]
Further if 
\[Z(A,B)\coloneqq \prod_{\substack{\alpha\in A,\\\beta\in B}}\frac{1}{(1-e^{-(\alpha+\beta)})},\]then define
\[Z(A,B;C,D)\coloneqq\frac{Z(A,B)Z(C,D)}{Z(A,D)Z(B,C)}.\] Finally, if $S\subset A$ and $T\subset B$ then $\overline{S}=A-S$, $\overline{T}=B-T$, $S^-=\{-\hat{\alpha}:\hat{\alpha}\in S\}$ and similarly for $T$.  Note that here we are using the notation $U+V$, $U-V$ (to be interpreted as $U\cup V$ and $U\backslash V$ respectively for sets $U, V$) to be consistent with the statement of the theorem in~\cite{consna08}.
\begin{thm}
(\cite{cfkrs1}) With $N\geq 0$ and $\RE(\gamma)>0, \RE(\delta)>0$ for $\gamma\in C$, $\delta\in D$, $|C|\leq |A|+N$, $|D|\leq |B|+N$, we have 
\[R(A,B;C,D)=\sum_{\substack{S\subset A,T\subset B\\|S|=|T|}}e^{-N(\sum_{\hat{\alpha}\in S}\hat{\alpha}+\sum_{\hat{\beta}\in T}\hat{\beta})}Z(\overline{S}+T^-,\overline{T}+S^-;C,D),\]
where $A=S+\overline{S}$ and $B=T+\overline{T}$. 
\end{thm}

To see how this is used to give the full polynomials for $\mom_N(k,\beta)$, we outline the simplest case with $k=\beta=1$. We note that 
\begin{align*}
\mom_N(1,1)&=\frac{1}{2\pi}\int_0^{2\pi}\mathbb{E}_{A\in U(N)}\left(|P_N(A,\theta)|^2\right)d\theta\\
&=\frac{1}{2\pi}\int_0^{2\pi}\int_{U(N)}P_N(A,\theta)P_N(A^*,-\theta)dAd\theta,
\end{align*}
so we apply the above theorem with $A=\{i\alpha\}, B=\{i\beta\}$, $C,D=\emptyset$.  This gives us that 
\begin{align*}
\mom_N(1,1)&=\frac{1}{2\pi}\int_0^{2\pi}\lim_{\beta\rightarrow-\alpha}Z(A,B)+e^{-iN(\alpha+\beta)}Z(B^-,A^-)d\alpha\\
&=\frac{1}{2\pi}\int_0^{2\pi}\lim_{\beta\rightarrow-\alpha}\sum_{m=0}^Ne^{-im(\alpha+\beta)}d\alpha\\
&=N+1.
\end{align*}

Higher values of $k, \beta$ clearly result in bigger sets $A, B$, and hence many more choices for $S, T$.  Nevertheless, small cases of $\mom_N(k,\beta)$ can be computed in the same way.  For example
\begin{align*}
\mom_N(1,1)&=N+1\\
\mom_N(2,1)&=\frac{1}{6}(N+3)(N+2)(N+1)\\
\mom_N(3,1)&=\frac{1}{2520}(N+5)(N+4)(N+3)(N+2)(N+1)(N^2+6N+21)\\
\mom_N(4,1)&=\frac{1}{778377600}(N+7)(N+6)(N+5)(N+4)(N+3)(N+2)(N+1)\\ 
&\times(7N^6+168N^5+1804N^4+10944N^3+41893N^2+99624N+154440)\\
\mom_N(1,2)&=\frac{1}{12}(N+1)(N+2)^2(N+3)\\
\mom_N(2,2)&=\frac{1}{163459296000}(N+7)(N+6)(N+5)(N+4)(N+3)(N+2)(N+1)\\
&\times(298N^8+9536N^7+134071N^6+1081640N^5+5494237 N^4+18102224N^3\\
&+38466354N^2+50225040N+32432400)\\
\mom_N(2,3)&=\frac{ (N+1) (N+2) (N+3) (N+4) (N+5) (N+6) (N+7) (N+8) (N+9) (N+10) (N+11)}{1722191327731024154944441889587200000000}\\
&\times \big(12308743625763
   N^{24}+1772459082109872 N^{23}+121902830804059138 N^{22}+\\
&+5328802119564663432 N^{21}+166214570195622478453 N^{20}+3937056259812505643352 N^{19}\\
&+73583663800226157619008 N^{18}+1113109355823972261429312 N^{17}+\\
&13869840005250869763713293
   N^{16}+144126954435929329947378912 N^{15}\\
&+1259786144898207172443272698
   N^{14}+9315726913410827893883025672 N^{13}\\
&+58475127984013141340467825323
   N^{12}+311978271286536355427593012632 N^{11}\\
&+1413794106539529439589778645028
   N^{10}+5427439874579682729570383266992 N^9\\
&+17564370687865211818995713096848
   N^8+47561382824003032731805262975232 N^7\\
&+106610927256886475209611301000128
   N^6+194861499503272627170466392014592 N^5\\
&+284303877221735683573377603640320
   N^4+320989495108428049992898521600000 N^3\\
&+266974288159876385845370793984000
   N^2+148918006780282798012340305920000
   N\\
&+43144523802785397500411904000000\big).
\end{align*}

It is worth noting explicitly that this method gives exact information about the moments of the moments at the freezing transition $\beta=1$.  

\subsection{Vandermonde Determinant Coefficients}\label{app:scott}

Recall that we are interested in determining the coefficient of terms of the form $(x_1\cdots x_n)^{n-1}$ in the square of the Vandermonde determinant,
\begin{equation}
\Delta(x_1,\dots,x_n)^2=\sum_{\sigma,\tau\in S_n}\Sgn(\sigma)\Sgn(\tau)\prod_{i=1}^nx_i^{\sigma(i)+\tau(i)-2}.
\end{equation}

Thus, we require that $\sigma(i)+\tau(i)=n+1$ for all $i\in\{1,\dots,n\}$ and in particular we want to show that this coefficient is strictly positive. 

Immediately, we see that there will be $n!$ terms of the required form since fixing $\sigma(i)$ completely determines $\tau(i)$. Consider the bijection 
\begin{align*}
\phi:\{1,\dots,n\}\rightarrow\{1,\dots,n\}, \quad i\mapsto n+1-i.
\end{align*}
The order of $\phi$ is 2 and if $n$ is even, it has no fixed point, whereas if $n$ is odd there is a unique fixed point $(n+1)/2$.  Thus, $\phi\in S_n$ and it consists of $n/2$ transpositions if $n$ is even, and $(n-1)/2$ transpositions if $n$ is odd. Now set $\tau=\phi\circ\sigma$, so $\tau\in S_n$, and $\tau(i)=n+1-\sigma(i)$.  Given $\sigma$, we have found our unique permutation.  To determine the sign of $\tau$, note that $\Sgn(\tau)=\Sgn(\phi)\Sgn(\sigma)$, and 
\[\Sgn(\phi)=(-1)^{\left\lfloor\tfrac{n}{2}\right\rfloor}=\begin{cases}+1&\text{if }n\equiv 0,1\mod 4\\ -1&\text{if }n\equiv 2,3\mod 4.\end{cases}\]
Thus, the coefficient of $(x_1\cdots x_n)^{n-1}$ in $\Delta(x_1,\dots,x_n)^2$ is $\Sgn(\phi)n!$. It now follows that the coefficient of 
\[\prod_{i=1}^{l_1}v_i^{l_1-1}\prod_{i=l_1+1}^{l_1+l_2}v_i^{l_2-1}\cdots \prod_{i=\sum_{j=1}^{k-1}l_j+1}^{k\beta}v_i^{l_k-1}\prod_{i=k\beta+1}^{\sum_{j=1}^{k-1}l_j+2\beta}v_i^{2\beta-l_k-1}\cdots\prod_{i=2(k-1)\beta+1+l_1}^{2k\beta}v_i^{2\beta-l_1-1}\] 
 in  
\[\prod_{n=1}^{k}\Delta(v_{\sum_{j=1}^{n-1}l_j+1},\dots,v_{\sum_{j=1}^{n}l_j})^2 \prod_{n=1}^{k}\Delta(v_{\sum_{j=1}^{n}l_j+2(k-n)\beta+1},\dots,v_{\sum_{j=1}^{n-1}l_j+2(k-(n-1))\beta})^2\]
is given by
\begin{align*}
(-1)^{\sum_{j=1}^k\left(\big\lfloor \tfrac{l_j}{2}\big\rfloor+\big\lfloor \tfrac{2\beta-l_j}{2}\big\rfloor\right)}\prod_{j=1}^kl_j!(2\beta-l_j)!&=(-1)^{k\beta+\sum_{j=1}^k\left(\big\lfloor \tfrac{l_j}{2}\big\rfloor+\big\lfloor \tfrac{-l_j}{2}\big\rfloor\right)}\prod_{j=1}^kl_j!(2\beta-l_j)!\\
&=(-1)^{k\beta}(-1)^{\sum_{j=1}^k\delta_{\{l_j\text{ is odd\}}}}\prod_{j=1}^kl_j!(2\beta-l_j)!\\
&=(-1)^{k\beta}(-1)^{\#\{j:l_j\text{ is odd}\}}\prod_{j=1}^kl_j!(2\beta-l_j)!.\end{align*}
This proves the result since the parity of $\#\{j:l_j\text{ is odd}\}$ is the same as the parity of $k\beta$ as $\sum_{j=1}^kl_j= k\beta$.  

\subsection{Uniqueness of Construction}\label{sec:uniqueconstruction}

When trying to construct the term of the form $(v_1\cdots v_{2k\beta})^{2\beta-1}$ in

\[\prod_{q=1}^{k}\Delta(v_{\sum_{j=1}^{q-1}l_j+1},\dots,v_{\sum_{j=1}^{q}l_j})^2 \prod_{q=1}^{k}\Delta(v_{\sum_{j=1}^{q}l_j+2(k-q)\beta+1},\dots,v_{\sum_{j=1}^{q-1}l_j+2(k-(q-1))\beta})^2\prod_{\substack{m\leq k\beta<n\\\alpha_m=\alpha_n}}\left({v_n-v_m}\right),\]
first note that the variables $v_m$, for $m\in\{1,\dots,k\beta\}$ only appear in the Vandermonde determinants and the products
\[\prod_{\substack{m\leq k\beta<n\\\alpha_m=\alpha_n}}\left({v_n-v_m}\right)=\prod_{\substack{m\in\{1,\dots,l_1\}\\n\in\{2(k-1)\beta+1+l_1,\dots,2k\beta\}}}(v_n-v_m)\cdots\prod_{\substack{m\in\{\sum_{j=1}^{k-1}l_j+1,\dots,k\beta\}\\n\in\{k\beta+1,\dots,\sum_{j=1}^{k-1}l_j+2\beta\}}}(v_n-v_m).\]

In particular, after fixing $q\in\{1,\dots,k\}$ take $v_j$ with $j\in\{\sum_{i=1}^{q-1}l_i+1,\dots,\sum_{i=1}^ql_i\}$. Then $v_j$ only appears in the following two terms: 
\[\Delta(v_{\sum_{i=1}^{q-1}l_i+1},\dots,v_{\sum_{i=1}^ql_i})^2\text{ and }\prod_{\substack{m\in\{\sum_{i=1}^{q-1}l_i+1,\dots,\sum_{i=1}^ql_i\}\\n\in\{2k\beta-\sum_{i=1}^{q}(2\beta-l_i)+1,\dots,2k\beta-\sum_{i=1}^{q-1}(2\beta-l_i)\}}}(v_n-v_m).\] In particular these are both homogeneous polynomials: the former of degree $l_q(l_q-1)$ in $l_q$ variables and the latter is of degree $l_q(2\beta-l_q)$ in $2\beta$ variables.  We will show that the only way to construct a term of the form $(v_1\cdots v_{2k\beta})^{2\beta-1}$ is as described following \cref{expansion}.  Without loss of generality, we will set $q=1$ and assume $l_1\geq 2$. From the above discussion, the square of the Vandermonde determinant consists of terms of the form 
\begin{equation}\label{aweights}
v_1^{a_1}\cdots v_{l_1}^{a_{l_1}},\text{ with }\sum_{i=1}^{l_1}a_i=l_1(l_1-1).
\end{equation}
Similarly, the product term is built of elements of the form 
\begin{equation}\label{bweights}
v_1^{b_1}\cdots v_{l_1}^{b_{l_1}}v_{2(k-1)\beta+1+l_1}^{b_{l_1+1}}\cdots v_{2k\beta}^{b_{2\beta}},\text{ with }\sum_{i=1}^{2\beta}b_i=l_1(2\beta-l_1),\ 0\leq b_i\leq 2\beta-l_1.
\end{equation}
Hence, each term of 
\[\Delta(v_1,\dots,v_{l_1})^2\prod_{\substack{m\in\{1,\dots,l_1\}\\n\in\{2(k-1)\beta+1+l_1,\dots,2k\beta\}}}(v_n-v_m)\]
is of the form 
\[v_1^{a_1+b_1}\cdots v_{l_1}^{a_{l_1}+b_{l_1}}v_{2(k-1)\beta+1+l_1}^{b_{l_1+1}}\cdots v_{2k\beta}^{b_{2\beta}},\]
with $a_i,b_i$ satisfying the homogenous conditions. To reach our goal, we need to find all possibilities for $a_i, 1\leq i\leq l_1$ and $b_i, 1\leq i\leq 2\beta$ that $a_i+b_i=2\beta-1$ for $i\in\{1,\dots,l_1\}$. This implies that we need $\sum_{i=1}^{l_1}(a_i+b_i)=l_1(2\beta-1)$.  Now note that the `homogeneous conditions' in \cref{aweights} and \cref{bweights} together mean that 
\[\sum_{i=1}^{l_1}(a_i+b_i)+\sum_{l_1+1}^{2\beta}b_i=\sum_{i=1}^{l_1}a_i+\sum_{i=1}^{2\beta}b_i=l_1(2\beta-1).\]
Thus, we must set $b_{l_1+1},\dots,b_{2\beta}=0$ if we want to construct the required term.  This leaves us with finding all $a_i,b_i$ $1\leq i\leq l_1$ such that all the following are satisfied,
\begin{align*}
a_i+b_i&=2\beta-1,\\
\sum_{i=1}^{l_1}a_i&=l_1(l_1-1),\\
\sum_{i=1}^{l_1}b_i&=l_1(2\beta-l_1),\\
0\leq &b_i\leq 2\beta-l_1.
\end{align*}
However, the latter two conditions imply that we must have $b_i=2\beta-l_1$ for all $1\leq i\leq l_1$ which in turn gives us that $a_i=l_1-1$ for all $1\leq i\leq l_1$, and these are the only possible choices. This exactly matches the construction described following \cref{expansion}.  The case for $q\in\{2,\dots,k\}$ follows similarly.

%%%%%%%%%%%%%%%%%%%
%                                                     %
%                                                     %
%             References                      %
%                                                     %
%                                                     %
%%%%%%%%%%%%%%%%%%%

\end{document}